\documentclass[conference]{IEEEtran}

\IEEEoverridecommandlockouts
\usepackage{float}

\usepackage{amsmath}
\usepackage{amssymb}
\usepackage{amsfonts}
\usepackage{amsthm}
\usepackage{mathtools}
\usepackage{bm}
\usepackage{mathrsfs}
\usepackage{xcolor}
\usepackage{graphicx}
\usepackage{float}
\usepackage{caption}
\usepackage{subcaption}
\usepackage{booktabs}
\usepackage{array}
\usepackage{multirow}
\usepackage{comment}
\usepackage{algorithmic}

\usepackage{cite}
\usepackage{url}

\usepackage{textcomp}
\usepackage{xcolor}

\usepackage{titlesec}
\usepackage{etoolbox}

\usepackage[
    colorlinks=true,
    linkcolor=blue,
    citecolor=blue,
    urlcolor=blue
]{hyperref}

\def\BibTeX{%
    {\rm B\kern-.05em%
    {\sc i\kern-.025em b}%
    \kern-.08em%
    T\kern-.1667em%
    \lower.7ex\hbox{E}%
    \kern-.125emX}%
}

\titlespacing*{\section}
{0pt}
{0.8pt plus 0.2pt minus 0.2pt}
{0.15pt}

\titlespacing*{\subsection}
{0pt}
{0.6pt plus 0.2pt minus 0.2pt}
{0.10pt}

\titlespacing*{\subsubsection}
{0pt}
{0.5pt plus 0.2pt minus 0.2pt}
{0.10pt}

\makeatletter
\g@addto@macro\@floatboxreset{%
    \setlength{\abovecaptionskip}{1pt}%
    \setlength{\belowcaptionskip}{0pt}%
}
\makeatother

\newtheoremstyle{compactthm}
{0.3pt}
{0.3pt}
{\itshape}
{0pt}
{\bfseries}
{.}
{0.30em}
{}

\newtheoremstyle{compactdef}
{0.3pt}
{0.3pt}
{\normalfont}
{0pt}
{\bfseries}
{.}
{0.30em}
{}

\theoremstyle{compactthm}

\newtheorem{theorem}{Theorem}
\newtheorem{lemma}{Lemma}

\theoremstyle{compactdef}

\newtheorem{assumption}{Assumption}
\newtheorem{definition}{Definition}

\BeforeBeginEnvironment{theorem}{\vspace{-1.5pt}}
\AfterEndEnvironment{theorem}{\vspace{-1.5pt}}

\BeforeBeginEnvironment{lemma}{\vspace{-1.5pt}}
\AfterEndEnvironment{lemma}{\vspace{-1.5pt}}

\BeforeBeginEnvironment{assumption}{\vspace{-1.5pt}}
\AfterEndEnvironment{assumption}{\vspace{-1.5pt}}

\BeforeBeginEnvironment{proposition}{\vspace{-1.5pt}}
\AfterEndEnvironment{proposition}{\vspace{-1.5pt}}

\BeforeBeginEnvironment{corollary}{\vspace{-1.5pt}}
\AfterEndEnvironment{corollary}{\vspace{-1.5pt}}

\BeforeBeginEnvironment{definition}{\vspace{-1.5pt}}
\AfterEndEnvironment{definition}{\vspace{-1.5pt}}

\BeforeBeginEnvironment{remark}{\vspace{-1.5pt}}
\AfterEndEnvironment{remark}{\vspace{-1.5pt}}

\BeforeBeginEnvironment{proof}{\vspace{-2.5pt}}
\AfterEndEnvironment{proof}{\vspace{-2.0pt}}

\AtBeginDocument{%
    \BeforeBeginEnvironment{IEEEproof}{\vspace{-2.5pt}}%
    \AfterEndEnvironment{IEEEproof}{\vspace{-2.0pt}}%
}

\BeforeBeginEnvironment{figure}{\vspace{-1pt}}
\AfterEndEnvironment{figure}{\vspace{-1.5pt}}

\BeforeBeginEnvironment{figure*}{\vspace{-1pt}}
\AfterEndEnvironment{figure*}{\vspace{-1.5pt}}

\BeforeBeginEnvironment{table}{\vspace{-1pt}}
\AfterEndEnvironment{table}{\vspace{-1.5pt}}

\BeforeBeginEnvironment{table*}{\vspace{-1pt}}
\AfterEndEnvironment{table*}{\vspace{-1.5pt}}

\newcommand{\bx}{\bm{x}}
\newcommand{\bu}{\bm{u}}
\newcommand{\bfv}{\bm{f}}
\newcommand{\bg}{\bm{g}}
\newcommand{\bW}{\bm{W}}
\newcommand{\bphi}{\bm{\phi}}
\newcommand{\bsigma}{\bm{\sigma}}
\newcommand{\bR}{\bm{R}}

\begin{document}


\title{
Event-Triggered Practical Fixed-Time Integral Reinforcement Learning
for Unknown Nonlinear Systems
}


\author{
Tien Dat Vu,
My Nguyen Bach,
and Minh Doan%
\thanks{
T. D. Vu and M. Doan are with the Faculty of Mechanical Engineering,
Ho Chi Minh City University of Technology (HCMUT),
Vietnam National University Ho Chi Minh City (VNU-HCM),
Ho Chi Minh City, Vietnam
(e-mails: dat.vuv@hcmut.edu.vn; minh.doan@hcmut.edu.vn).
}%
\thanks{
M. N. Bach is with the Department of Mathematics and Statistics,
University of New Mexico,
Albuquerque, NM, USA
(e-mail: ngocmynguyenbach@gmail.com).
}%
\thanks{
Corresponding author: Minh Doan
(e-mail: minh.doan@hcmut.edu.vn).
}
}

\maketitle


\begin{abstract}
This paper develops an event-triggered fixed-time integral reinforcement learning framework for optimal control of unknown nonlinear systems. An integral data-driven identifier is first used to reconstruct the unknown dynamics, after which an inverse-optimal formulation is employed to construct a fixed-time running cost. A learning law satisfying the practical fixed-time property is then derived. Previously collected data, or data obtained during a finite excitation interval, are stored in an experience-replay buffer and incorporated into the weight-update law. This avoids the persistent-excitation condition, which is often difficult to satisfy in practical operation. To reduce communication and control updates, an event-triggered mechanism is introduced. The paper shows that, under the event-triggered implementation, the closed-loop system still achieves practical fixed-time stability, while the proposed triggering rule guarantees the exclusion of Zeno behavior. Finally, a nonlinear example is presented to verify the theoretical results developed in the paper.
\end{abstract}




\section{Introduction}
\label{sec:introduction}

Modern cyber--physical systems (CPSs) are becoming increasingly complex
in their structure, dynamics, and interaction between physical,
computational, and communication components, making an accurate
first-principles model increasingly difficult to obtain
\cite{Wan2021CPS,Lee2015CPS}. At the same time, stricter timing,
performance, and communication requirements have become important
design considerations in networked control systems
\cite{Lee2015CPS,Wan2021CPS}. Optimal control provides a systematic
framework for handling such demanding performance objectives
\cite{Wan2021CPS}. However, nonlinear optimal control is fundamentally
limited by the Hamilton--Jacobi--Bellman (HJB) equation, which generally
does not admit an explicit analytical solution, while unknown system
dynamics further prevent direct construction of the optimal feedback law
\cite{Doya2000,Vamvoudakis2010}. When the dynamics are unknown,
guaranteeing convergence within a uniform time bound independent of the
initial condition becomes particularly difficult
\cite{Polyakov2012,Kokolakis2025Predefined}. Moreover, when sensing and
actuation are connected through a communication network, reducing
unnecessary control transmissions must also be considered without
sacrificing closed-loop performance
\cite{Lee2015CPS,Han2024ETIRL}.

Reinforcement learning (RL) and adaptive dynamic programming (ADP) have emerged
as data-driven approaches for approximating the solution of the HJB
equation without solving it analytically
\cite{Doya2000,Vrabie2009,Vamvoudakis2010}. In particular, integral
reinforcement learning (IRL) replaces the differential Bellman equation
with an integral relation constructed from measured trajectories, thereby
reducing the dependence of the learning process on unknown drift dynamics
\cite{Modares2014}. In parallel, fixed-time and predefined-time optimal
stabilization have provided a systematic way to impose convergence-time
requirements that are independent of the initial condition
\cite{Polyakov2012,Jimenez2017Predefined,Kokolakis2025Predefined}.
Kokolakis \emph{et al.} further developed fixed- and predefined-time
critic-learning laws using current and recorded data, showing that
sufficiently rich stored information can replace the persistent-excitation
requirement
\cite{Kokolakis2023FixedTime,Kokolakis2025Predefined}. These ideas also
motivate an inverse-optimal viewpoint, in which a prescribed Lyapunov
decay is used to construct a compatible running cost and optimal feedback
law
\cite{Jimenez2017Predefined,Kokolakis2025Predefined}. For an unknown
system, however, such a construction first requires the drift
information; in this work, the drift is therefore approximated by a
neural basis representation and its unknown weights are recovered from
integral data through a finite-data least-squares problem under a
full-row-rank condition (See Eq.~\eqref{eq:weight_identification}).

Several related developments address only parts of this problem.
Modares \emph{et al.} combined IRL with experience replay for partially
unknown continuous-time systems and used recorded data to relax the
persistent-excitation requirement
\cite{Modares2014}. Han \emph{et al.} developed event-triggered IRL for
unknown constrained nonlinear systems by combining system identification,
experience replay, and event-triggered actor--critic learning, but
fixed-time critic convergence was not established
\cite{Han2024ETIRL}. Hong \emph{et al.} combined IRL with dynamic
event-triggered adaptive dynamic programming for partially unknown
nonlinear systems and excluded Zeno behavior, while the resulting
closed-loop and learning errors were characterized through robust or
ultimate-bounded stability rather than fixed-time convergence
\cite{Hong2025DETADP}. Zhang \emph{et al.} incorporated fixed-time
prescribed performance into event-triggered critic-only optimal tracking,
but the neural and closed-loop signals were again established as uniformly
ultimately bounded
\cite{Zhang2026ETFixedTime}. Kokolakis \emph{et al.} achieved
predefined-time critic learning and optimal stabilization using current
and replayed data; however, the HJB residual used by the learning law
depends explicitly on the system dynamics, and event-triggered
communication was not considered
\cite{Kokolakis2025Predefined}.

These results leave a gap between time-critical optimal learning,
unknown-dynamics reconstruction, and communication-efficient
implementation. In particular, an inverse-optimal fixed-time construction
cannot be directly implemented when the drift required by the optimality
conditions is unavailable, whereas existing event-triggered IRL schemes
for unknown nonlinear systems do not simultaneously provide practical
fixed-time critic convergence and practical fixed-time closed-loop
stability
\cite{Han2024ETIRL,Hong2025DETADP,Kokolakis2025Predefined,
Zhang2026ETFixedTime}.

To address this gap, this paper contributes an event-triggered fixed-time IRL
framework for unknown nonlinear systems. First, the unknown drift is
reconstructed from integral data through a neural approximation and a
finite-data least-squares identifier, after which the identified model is
used to construct an inverse-optimal running cost associated with a
designer-selected value function and fixed-time decay law. Second, a
critic-only IRL law combines current integral data with informative
experience-replay data, avoids the persistent-excitation requirement, and
establishes practical fixed-time convergence of the critic weights while
retaining the value-function approximation residual. Third, an
event-triggered sample-and-hold implementation is coupled with the
learning dynamics, and a composite Lyapunov analysis establishes practical
fixed-time closed-loop stability together with a strictly positive
inter-event-time bound that excludes Zeno behavior.


\section{Problem Formulation and Data-Driven System Identification}
\label{sec:problem_identification}

\subsection{Problem Formulation}




Consider the nonlinear control-affine system
\begin{equation}
    \dot{\bx}
    =
    \bfv(\bx)
    +
    \bg(\bx)\bu,
    \label{eq:plant}
\end{equation}
where \(\bx\in\mathbb R^n\) is the state, \(\bu\in\mathbb R^m\) is the control input, \(\bfv:\mathbb R^n\rightarrow\mathbb R^n\) is an unknown nonlinear drift, and \(\bg:\mathbb R^n\rightarrow\mathbb R^{n\times m}\) is known. The objective is to learn the unknown drift from measured state--input data and construct an optimal feedback policy that drives the state to the origin within a uniform finite settling time.

\begin{definition}[Admissible policy~\cite{AbuKhalaf2005}]
\label{def:admissible_secure_policy}
A continuous policy \(\bu_c:\Omega\rightarrow\mathcal U\) is admissible
on \(\Omega\), denoted by \(\bu_c\in\Psi(\Omega)\), if
\(\bu_c(\bm0)=\bm0\), the resulting closed-loop system is
asymptotically stable on \(\Omega\) in the absence of exogenous signals,
and the associated infinite-horizon cost is finite
(See Eq.~\eqref{eq:cost_function}).
\end{definition}

\begin{assumption}
\label{ass:plant}
Let \(\Omega\subset\mathbb R^n\) be compact and contain the origin in its interior. The unknown drift \(\bfv(\bx)\) and known input matrix \(\bg(\bx)\) are locally Lipschitz on an open neighborhood of \(\Omega\), with \(\bfv(\bm0)=\bm0\). For every admissible input considered herein, the corresponding solution of \eqref{eq:plant} exists uniquely while it remains in \(\Omega\).
\end{assumption}

\begin{lemma}[Fixed-time stability~\cite{Polyakov2012}]
\label{lem:fixed_time}
Let \(V(\bx)\) be a positive-definite continuously differentiable function satisfying \(\dot V\le-aV^\mu-bV^\nu\), where \(a>0\), \(b>0\), and \(0<\mu<1<\nu\). Then, the origin is fixed-time stable and its settling time satisfies
\begin{equation}
    T_{\max}
    \le
    \frac{1}{a(1-\mu)}
    +
    \frac{1}{b(\nu-1)}.
    \label{eq:fixed_time_bound}
\end{equation}
\end{lemma}

\subsection{Neural Approximation of the Unknown Dynamics}

Since \(\bfv(\bx)\) is unavailable for control design, it is first reconstructed from measured data. By the universal approximation property~\cite{Kokolakis2025Predefined}, each continuous drift on the compact set \(\Omega\) can be represented as
\begin{equation}
    \bfv(\bx)
    =
    \bW_f^{*\top}\bsigma_f(\bx)
    +
    \bm{\varepsilon}_f(\bx),
    \label{eq:f_approximation}
\end{equation}
where \(\bsigma_f(\bx)\in\mathbb R^{N_f}\) is a prescribed vector of continuously differentiable basis functions satisfying
\(
\bsigma_f(\bm 0)=\bm 0
\), \(\bW_f^*\in\mathbb R^{N_f\times n}\) is the unknown ideal weight matrix, and \(\bm{\varepsilon}_f(\bx)\in\mathbb R^n\) is the approximation error. On \(\Omega\), \(\|\bsigma_f(\bx)\|\le\bar\sigma_f\) and \(\|\bm{\varepsilon}_f(\bx)\|\le\bar\varepsilon_f\) for some finite constants \(\bar\sigma_f>0\) and \(\bar\varepsilon_f>0\).

The identified drift is subsequently represented by
\(
\hat{\bfv}(\bx)
=
\hat{\bW}_f^\top\bsigma_f(\bx)
\),
where
\(
\hat{\bW}_f
\)
denotes the estimated weight matrix associated with the prescribed basis
vector
\(
\bsigma_f(\bx)
\).

\subsection{Integral Data-Driven Identification}
\label{subsec:integral_identification}

Let
\(
\tau_0<\tau_1<\cdots
\)
denote the sampling instants used to construct integral identification
data. Substituting the approximation of the unknown drift
\(
\bfv(\bx)
=
\bW_f^{*\top}\bsigma_f(\bx)
+
\bm{\varepsilon}_f(\bx)
\)
from \eqref{eq:f_approximation} into \eqref{eq:plant} and integrating over
\(
[\tau_\ell,\tau_{\ell+1}]
\)
gives
\(
\bx(\tau_{\ell+1})-\bx(\tau_\ell)
-
\int_{\tau_\ell}^{\tau_{\ell+1}}
\bg(\bx(\tau))\bu(\tau)\,d\tau
=
\bW_f^{*\top}
\int_{\tau_\ell}^{\tau_{\ell+1}}
\bsigma_f(\bx(\tau))\,d\tau
+
\int_{\tau_\ell}^{\tau_{\ell+1}}
\bm{\varepsilon}_f(\bx(\tau))\,d\tau
\).

For each interval, define
\(
\bm y_{f,\ell}
:=
\bx(\tau_{\ell+1})
-
\bx(\tau_\ell)
-
\int_{\tau_\ell}^{\tau_{\ell+1}}
\bg(\bx(\tau))\bu(\tau)\,d\tau
\)
and
\(
\bm\phi_{f,\ell}
:=
\int_{\tau_\ell}^{\tau_{\ell+1}}
\bsigma_f(\bx(\tau))\,d\tau
\).
Define
\(
\bm d_{f,\ell}
:=
\int_{\tau_\ell}^{\tau_{\ell+1}}
\bm{\varepsilon}_f(\bx(\tau))\,d\tau
\).
From \eqref{eq:f_approximation}, each data segment satisfies
\(
\bm y_{f,\ell}
=
\bW_f^{*\top}\bm\phi_{f,\ell}
+
\bm d_{f,\ell}
\).

Collecting
\(
M_f
\)
data segments gives
\begin{equation}
    \bm Y_f
    =
    \bW_f^{*\top}\bm\Phi_f
    +
    \bm D_f,
    \label{eq:stacked_identification}
\end{equation}
where
\(
\bm Y_f
:=
[\bm y_{f,0},\ldots,\bm y_{f,M_f-1}]
\),
\(
\bm\Phi_f
:=
[\bm\phi_{f,0},\ldots,\bm\phi_{f,M_f-1}]
\),
and
\(
\bm D_f
:=
[\bm d_{f,0},\ldots,\bm d_{f,M_f-1}]
\).
The estimate
\(
\hat{\bW}_f
\)
is obtained by fitting
\(
\bW_f^\top\bm\Phi_f
\)
to the accumulated data
\(
\bm Y_f
\)
in the least-squares sense.

To measure the information accumulated from the collected data, define
the cumulative information matrix
\(
\bm{\mathcal S}_f(M)
:=
\sum_{\ell=0}^{M-1}
\bm\phi_{f,\ell}
\bm\phi_{f,\ell}^{\top}
\).
When a new data segment becomes available,
\(
\bm{\mathcal S}_f(M+1)
=
\bm{\mathcal S}_f(M)
+
\bm\phi_{f,M}\bm\phi_{f,M}^{\top}
\).
At
\(
M=M_f
\),
one has
\(
\bm{\mathcal S}_f(M_f)
=
\bm\Phi_f\bm\Phi_f^\top
\).

\begin{assumption}[Finite excitation]
\textcolor{red}{\label{ass:finite_excitation}}
The collected data are sufficiently informative such that there exist
a finite integer
\(
M_f\geq N_f
\)
and a constant
\(
\alpha_f>0
\)
satisfying
\(
\bm{\mathcal S}_f(M_f)
=
\sum_{\ell=0}^{M_f-1}
\bm\phi_{f,\ell}\bm\phi_{f,\ell}^{\top}
\succeq
\alpha_f\bm I
\).
Equivalently,
\(
\operatorname{rank}(\bm\Phi_f)=N_f
\).
\end{assumption}


The unknown weight matrix is estimated from the accumulated data by
solving the least-squares problem

\(
\hat{\bW}_f
=
\arg\min_{\bW_f}
\frac{1}{2}
\left\|
\bm Y_f-\bW_f^\top\bm\Phi_f
\right\|_F^2
\).
The corresponding normal equation is
\(
\left(
\bm Y_f-\hat{\bW}_f^\top\bm\Phi_f
\right)
\bm\Phi_f^\top
=
\bm0
\),
or equivalently,
\(
\hat{\bW}_f^\top
\bm\Phi_f\bm\Phi_f^\top
=
\bm Y_f\bm\Phi_f^\top
\).
Since
\(
\bm\Phi_f\bm\Phi_f^\top
\)
is nonsingular under Assumption~\ref{ass:finite_excitation}, the unique least-squares minimizer is

\begin{equation}
\begin{aligned}
    \hat{\bW}_f^\top
    &=
    \bm Y_f
    \bm\Phi_f^\top
    \left(
        \bm\Phi_f\bm\Phi_f^\top
    \right)^{-1}
    \\
    &=
    \left(
        \sum_{\ell=0}^{M_f-1}
        \bm y_{f,\ell}
        \bm\phi_{f,\ell}^{\top}
    \right)
    \left(
        \sum_{\ell=0}^{M_f-1}
        \bm\phi_{f,\ell}
        \bm\phi_{f,\ell}^{\top}
    \right)^{-1}.
\end{aligned}
\label{eq:weight_identification}
\end{equation}

\begin{lemma}[Bounded identification error]
\label{lem:identification_error}
Under Assumptions~\ref{ass:plant} and
\ref{ass:finite_excitation}, the identification error is uniformly
bounded on
\(
\Omega
\).
In particular, there exists a finite constant
\(
\Delta_f>0
\)
such that
\begin{equation}
    \left\|
        \bfv(\bx)-\hat{\bfv}(\bx)
    \right\|
    \leq
    \Delta_f,
    \qquad
    \forall\,\bx\in\Omega.
    \label{eq:identification_error_bound}
\end{equation}
\end{lemma}

\begin{proof}
From \eqref{eq:stacked_identification} and
\eqref{eq:weight_identification}, the full-row-rank property of
\(
\bm\Phi_f
\)
gives
\(
\hat{\bW}_f^\top
=
\bW_f^{*\top}
+
\bm D_f\bm\Phi_f^\top
\bm{\mathcal S}_f^{-1}(M_f)
\).
Hence,
\(
\bfv(\bx)-\hat{\bfv}(\bx)
=
-
\bm D_f\bm\Phi_f^\top
\bm{\mathcal S}_f^{-1}(M_f)
\bsigma_f(\bx)
+
\bm{\varepsilon}_f(\bx)
\).
Using the boundedness of
\(
\bsigma_f(\bx)
\)
and
\(
\bm{\varepsilon}_f(\bx)
\)
on
\(
\Omega
\),
one obtains
\(
\left\|
\bfv(\bx)-\hat{\bfv}(\bx)
\right\|
\leq
\left\|
\bm D_f\bm\Phi_f^\top
\bm{\mathcal S}_f^{-1}(M_f)
\right\|
\bar\sigma_f
+
\bar\varepsilon_f
\).
Therefore,
\eqref{eq:identification_error_bound}
holds with
\(
\Delta_f
:=
\left\|
\bm D_f\bm\Phi_f^\top
\bm{\mathcal S}_f^{-1}(M_f)
\right\|
\bar\sigma_f
+
\bar\varepsilon_f
\).
\end{proof}

\section{Fixed-Time Reinforcement Learning Optimal Control}
\label{sec:fixed_time_rl}

\subsection{Fixed-Time Optimal-Control Formulation}

Let \(V_d:\Omega\rightarrow\mathbb R_{\geq0}\) be a continuously differentiable positive-definite function satisfying \(V_d(\bm 0)=0\). The desired closed-loop behavior is prescribed by
\begin{equation}
    \dot V_d(\bx)
    =
    -aV_d^{\gamma_1}(\bx)
    -
    bV_d^{\gamma_2}(\bx),
    \qquad
    a,b>0,\quad 0<\gamma_1<1<\gamma_2.
    \label{eq:desired_fixed_time_decay}
\end{equation}

The function \(V_d\) is freely selected by the designer to encode the desired closed-loop behavior. 
According to Lemma~\ref{lem:fixed_time}, any feedback law satisfying \eqref{eq:desired_fixed_time_decay} renders the origin fixed-time stable with a settling-time upper bound independent of the initial condition.

Consider the performance index

\begin{equation}
\begin{aligned}
    J(\bx_0,\bu)
    ={}&
    \int_0^\infty
    \left[
        \bm L_1(\bx)\bu
        +
        L_2(\bx)
        +
        \bu^\top\bR\bu
    \right]dt,
    \\
    &\qquad
    \bR=\bR^\top>0.
\end{aligned}
\label{eq:cost_function}
\end{equation}

where \(\bm L_1:\Omega\rightarrow\mathbb R^{1\times m}\) and \(L_2:\Omega\rightarrow\mathbb R\) are to be constructed such that \(V_d\) is the optimal value function and the corresponding optimal policy satisfies \eqref{eq:desired_fixed_time_decay}.

For the identified model \(\dot{\bx}=\hat{\bfv}(\bx)+\bg(\bx)\bu\), the Hamiltonian associated with \eqref{eq:cost_function} is
\begin{equation}
\begin{aligned}
    H(\bx,\bu,\nabla V_d)
    ={}&
    \bm L_1(\bx)\bu
    +
    L_2(\bx)
    +
    \bu^\top\bR\bu
    \\
    &+
    \nabla V_d^\top(\bx)\hat{\bfv}(\bx)
    +
    \nabla V_d^\top(\bx)\bg(\bx)\bu .
\end{aligned}
\label{eq:hamiltonian_fixed_time}
\end{equation}
Equivalently, \(H=L_2+\nabla V_d^\top\hat{\bfv}+\bu^\top\bR\bu+\left(\bm L_1+\nabla V_d^\top\bg\right)\bu\).

The HJB equation is
\begin{equation}
    \min_{\bu}
    H(\bx,\bu,\nabla V_d)
    =
    0.
    \label{eq:hjb_fixed_time}
\end{equation}
Since \(\bR>0\), \(\frac{\partial^2 H}{\partial\bu^2}=2\bR>0\), and therefore the unique minimizing policy satisfies \(\frac{\partial H}{\partial\bu}=\bm L_1^\top(\bx)+2\bR\bu+\bg^\top(\bx)\nabla V_d(\bx)=\bm 0\). Hence, \(2\bR\bu^*=-\left[\bm L_1^\top(\bx)+\bg^\top(\bx)\nabla V_d(\bx)\right]\), which gives
\begin{equation}
    \bu^*(\bx)
    =
    -\frac{1}{2}\bR^{-1}
    \left[
        \bm L_1^\top(\bx)
        +
        \bg^\top(\bx)\nabla V_d(\bx)
    \right].
    \label{eq:optimal_control_general}
\end{equation}

\subsection{Reverse Construction of the Fixed-Time Optimal Cost}

Along \(\dot{\bx}=\hat{\bfv}(\bx)+\bg(\bx)\bu\), the derivative of \(V_d\) is \(\dot V_d=\nabla V_d^\top(\bx)\dot{\bx}=\nabla V_d^\top(\bx)\hat{\bfv}(\bx)+\nabla V_d^\top(\bx)\bg(\bx)\bu\). Imposing \eqref{eq:desired_fixed_time_decay} gives \(\nabla V_d^\top(\bx)\hat{\bfv}(\bx)+\nabla V_d^\top(\bx)\bg(\bx)\bu=-aV_d^{\gamma_1}(\bx)-bV_d^{\gamma_2}(\bx)\), or equivalently, \(\nabla V_d^\top(\bx)\bg(\bx)\bu=-\nabla V_d^\top(\bx)\hat{\bfv}(\bx)-aV_d^{\gamma_1}(\bx)-bV_d^{\gamma_2}(\bx)\).

Define \(\sigma_d(\bx):=\nabla V_d^\top(\bx)\hat{\bfv}(\bx)+aV_d^{\gamma_1}(\bx)+bV_d^{\gamma_2}(\bx)\) and \(D_d(\bx):=\nabla V_d^\top(\bx)\bg(\bx)\bR^{-1}\bg^\top(\bx)\nabla V_d(\bx)\). Thus, the required condition becomes \(\nabla V_d^\top(\bx)\bg(\bx)\bu=-\sigma_d(\bx)\).


Since infinitely many control inputs \(\bu\) may satisfy
\(
\nabla V_d^\top(\bx)\bg(\bx)\bu
=
-\sigma_d(\bx)
\),
we select the one minimizing the weighted control effort
\(
\bu^\top\bR\bu
\).
Accordingly, consider
\(
\min_{\bu}\;\bu^\top\bR\bu
\)
subject to
\(
\nabla V_d^\top(\bx)\bg(\bx)\bu
+
\sigma_d(\bx)
=
0
\).
The corresponding Lagrangian is
\(
\mathcal L_c(\bu,\lambda_c)
=
\bu^\top\bR\bu
+
\lambda_c
[
\nabla V_d^\top(\bx)\bg(\bx)\bu
+
\sigma_d(\bx)
]
\).
The stationarity condition with respect to \(\bu\) gives
\(
\partial\mathcal L_c/\partial\bu
=
2\bR\bu
+
\lambda_c
\bg^\top(\bx)\nabla V_d(\bx)
=
\bm0
\),
and therefore
\(
\bu
=
-\frac{\lambda_c}{2}
\bR^{-1}
\bg^\top(\bx)\nabla V_d(\bx)
\).
Substituting this expression into the constraint yields
\(
-\frac{\lambda_c}{2}
\nabla V_d^\top(\bx)
\bg(\bx)\bR^{-1}\bg^\top(\bx)
\nabla V_d(\bx)
+
\sigma_d(\bx)
=
0
\).
Using
\(
D_d(\bx)
:=
\nabla V_d^\top(\bx)
\bg(\bx)\bR^{-1}\bg^\top(\bx)
\nabla V_d(\bx)
\),
one obtains
\(
\lambda_c
=
2\sigma_d(\bx)/D_d(\bx)
\).
Hence, the unique \(\bR\)-weighted minimum-effort control satisfying the
prescribed fixed-time decay is


\begin{equation}
\begin{aligned}
    \bu_c(\bx)
    ={}&
    -
    \frac{
        \nabla V_d^\top(\bx)\hat{\bfv}(\bx)
        +
        aV_d^{\gamma_1}(\bx)
        +
        bV_d^{\gamma_2}(\bx)
    }{
        \nabla V_d^\top(\bx)
        \bg(\bx)\bR^{-1}\bg^\top(\bx)
        \nabla V_d(\bx)
    }
    \\
    &\qquad\times
    \bR^{-1}
    \bg^\top(\bx)
    \nabla V_d(\bx).
\end{aligned}
\label{eq:continuous_fixed_time_controller}
\end{equation}

The following condition ensures that \eqref{eq:continuous_fixed_time_controller} is well defined.

\begin{assumption}[Inverse-optimal feasibility]
\label{ass:inverse_optimal_feasibility}
For every \(\bx\in\Omega\setminus\{\bm 0\}\), \(\nabla V_d^\top(\bx)\bg(\bx)\bR^{-1}\bg^\top(\bx)\nabla V_d(\bx)>0\).
\end{assumption}

Define
\(
\bu_c(\bm 0)=\bm 0
\).
Since
\(
\bsigma_f(\bm 0)=\bm 0
\),
the identified drift satisfies
\(
\hat{\bfv}(\bm 0)=\bm 0
\).
Hence,
\[
\hat{\bfv}(\bm 0)+\bg(\bm 0)\bu_c(\bm 0)=\bm 0,
\]
so the origin is an equilibrium of the identified closed-loop system.


For every
\(
\bx\in\Omega\setminus\{\bm 0\}
\),
\[
\nabla V_d^\top\bg\bu_c
=
-\frac{\sigma_d}{D_d}
\nabla V_d^\top\bg\bR^{-1}\bg^\top\nabla V_d
=
-\sigma_d.
\] 

Therefore, for every
\(
\bx\in\Omega\setminus\{\bm 0\}
\),
\[
\dot V_d
=
\nabla V_d^\top\hat{\bfv}
+
\nabla V_d^\top\bg\bu_c
=
\nabla V_d^\top\hat{\bfv}
-
\sigma_d
=
-aV_d^{\gamma_1}-bV_d^{\gamma_2}.
\]

At \(\bx=\bm 0\), both sides are zero.

For every
\(
\bx\in\Omega\setminus\{\bm 0\}
\),
we now construct \(\bm L_1(\bx)\) such that the Hamiltonian-minimizing policy \(\bu^*\) coincides with \(\bu_c\).
From \eqref{eq:optimal_control_general},
\(
\bu^*
=
-\frac{1}{2}\bR^{-1}
\left[
\bm L_1^\top+\bg^\top\nabla V_d
\right]
\),
whereas
\(
\bu_c
=
-\frac{\sigma_d}{D_d}
\bR^{-1}\bg^\top\nabla V_d
\).
Imposing
\(
\bu^*=\bu_c
\)
gives
\(
-\frac{1}{2}\bR^{-1}
\left[
\bm L_1^\top+\bg^\top\nabla V_d
\right]
=
-\frac{\sigma_d}{D_d}
\bR^{-1}\bg^\top\nabla V_d
\).
Multiplying both sides by
\(
-2\bR
\)
yields
\(
\bm L_1^\top+\bg^\top\nabla V_d
=
\frac{2\sigma_d}{D_d}
\bg^\top\nabla V_d
\),
so that
\(
\bm L_1^\top
=
\left(
\frac{2\sigma_d}{D_d}-1
\right)
\bg^\top\nabla V_d
\).

For every
\(
\bx\in\Omega\setminus\{\bm 0\}
\),
expanding
\(
\sigma_d
\)
and
\(
D_d
\)
gives

\(
\begin{aligned}
\bm L_1^\top(\bx)
={}&
\left[
\frac{
2\left[
\nabla V_d^\top(\bx)\hat{\bfv}(\bx)
+
aV_d^{\gamma_1}(\bx)
+
bV_d^{\gamma_2}(\bx)
\right]
}{
\nabla V_d^\top(\bx)
\bg(\bx)\bR^{-1}
\bg^\top(\bx)
\nabla V_d(\bx)
}
-1
\right]
\\
&\qquad\times
\bg^\top(\bx)\nabla V_d(\bx)
\end{aligned}
\).

It remains to determine
\(
L_2(\bx)
\).
Substitution of
\(
\bu^*
=
-\frac{1}{2}\bR^{-1}
\left[
\bm L_1^\top+\bg^\top\nabla V_d
\right]
\)
into the Hamiltonian gives the minimized Hamiltonian
\(
H^*
=
L_2
+
\nabla V_d^\top\hat{\bfv}
-
\frac{1}{4}
\left[
\bm L_1^\top+\bg^\top\nabla V_d
\right]^\top
\bR^{-1}
\left[
\bm L_1^\top+\bg^\top\nabla V_d
\right]
\).
Using
\(
\bm L_1^\top+\bg^\top\nabla V_d
=
\frac{2\sigma_d}{D_d}
\bg^\top\nabla V_d
\),
the quadratic term becomes
\(
\frac{1}{4}
\left[
\bm L_1^\top+\bg^\top\nabla V_d
\right]^\top
\bR^{-1}
\left[
\bm L_1^\top+\bg^\top\nabla V_d
\right]
=
\frac{1}{4}
\left(
\frac{2\sigma_d}{D_d}
\right)^2
\nabla V_d^\top
\bg
\bR^{-1}
\bg^\top
\nabla V_d
=
\frac{\sigma_d^2}{D_d^2}
D_d
=
\frac{\sigma_d^2}{D_d}
\).
Hence,
\(
H^*
=
L_2
+
\nabla V_d^\top\hat{\bfv}
-
\frac{\sigma_d^2}{D_d}
\).
The HJB condition
\(
H^*=0
\)
gives
\(
L_2
=
-\nabla V_d^\top\hat{\bfv}
+
\frac{\sigma_d^2}{D_d}
\).
Therefore,
\(
\begin{aligned}
L_2(\bx)
={}&
-\nabla V_d^\top(\bx)\hat{\bfv}(\bx)
\\
&+
\frac{
\left[
\nabla V_d^\top(\bx)\hat{\bfv}(\bx)
+
aV_d^{\gamma_1}(\bx)
+
bV_d^{\gamma_2}(\bx)
\right]^2
}{
\nabla V_d^\top(\bx)
\bg(\bx)\bR^{-1}
\bg^\top(\bx)
\nabla V_d(\bx)
}
\end{aligned}
\).

.

\subsection{Critic-Only Integral Reinforcement Learning}

After constructing \(\bm L_1(\bx)\) and \(L_2(\bx)\), define
\(\ell(\bx,\bu):=\bm L_1(\bx)\bu+L_2(\bx)+\bu^\top\bR\bu\).
The optimal value function is
\begin{equation}
    V^*(\bx_0)
    =
    \min_{\bu}
    \int_0^\infty
    \ell(\bx(t),\bu(t))\,dt .
    \label{eq:optimal_value_rl}
\end{equation}

For
\(\dot{\bx}=\bfv(\bx)+\bg(\bx)\bu\), the Hamiltonian is
\(
H(\bx,\bu,\nabla V^*)
=
\ell(\bx,\bu)
+
\nabla V^{*\top}(\bx)
[\bfv(\bx)+\bg(\bx)\bu]
\).
The HJB equation satisfies
\(
0=\min_{\bu}H(\bx,\bu,\nabla V^*)
\).
Since
\(
\partial^2H/\partial\bu^2=2\bR>0
\),
the minimizing policy follows from
\(
\partial H/\partial\bu
=
\bm L_1^\top(\bx)
+
2\bR\bu
+
\bg^\top(\bx)\nabla V^*(\bx)
=
\bm0
\).
Hence,
\begin{equation}
    \bu^*(\bx)
    =
    -\frac{1}{2}
    \bR^{-1}
    \left[
        \bm L_1^\top(\bx)
        +
        \bg^\top(\bx)\nabla V^*(\bx)
    \right].
    \label{eq:critic_only_optimal_policy}
\end{equation}

On the compact set \(\Omega\), let
\begin{equation}
    V^*(\bx)
    =
    \bW^{*\top}\bphi(\bx)
    +
    \varepsilon(\bx),
    \label{eq:critic_value_approximation}
\end{equation}
where
\(\bphi(\bx)\in\mathbb R^{N_c}\),
\(\bW^*\in\mathbb R^{N_c}\), and
\(\varepsilon(\bx)\) denotes the approximation error. Then
\(
\nabla V^*(\bx)
=
(\nabla\bphi(\bx))^\top\bW^*
+
\nabla\varepsilon(\bx)
\).
The critic estimate is
\begin{equation}
    \hat V(\bx)
    =
    \hat{\bW}^{\top}\bphi(\bx),
    \qquad
    \nabla\hat V(\bx)
    =
    (\nabla\bphi(\bx))^\top\hat{\bW}.
    \label{eq:critic_estimate}
\end{equation}
Thus, the critic-generated policy is
\begin{equation}
    \hat{\bu}(\bx)
    =
    -\frac{1}{2}
    \bR^{-1}
    \left[
        \bm L_1^\top(\bx)
        +
        \bg^\top(\bx)
        (\nabla\bphi(\bx))^\top
        \hat{\bW}
    \right].
    \label{eq:critic_only_policy}
\end{equation}




Let \(\boldsymbol\mu(\bx)\) denote the target policy and let
\(\bu_b(t)\) denote the behavior input used to generate the measured
trajectory. Since the controller is implemented in an event-triggered
sample-and-hold manner, the actual behavior input satisfies
\(
\bu_b(t)
=
\hat{\bu}(\bx(t_k),\hat{\bm W}(t_k))
\)
for
\(
t\in[t_k,t_{k+1})
\),
whereas the continuously evaluated target policy is
\(
\boldsymbol\mu(\bx(t))
=
\hat{\bu}(\bx(t),\hat{\bm W}(t))
\).
Thus,
\(
\bu_b(t)-\boldsymbol\mu(\bx(t))
=
\bm e_u(t)
\).

Then
\(
\dot{\bx}
=
\bfv(\bx)+\bg(\bx)\bu_b
\),
whereas the Bellman equation associated with
\(\boldsymbol\mu\) is
\(
0=
\ell(\bx,\boldsymbol\mu)
+
\nabla V^\top(\bx)
[\bfv(\bx)+\bg(\bx)\boldsymbol\mu]
\).
Hence,
\(
\nabla V^\top
[\bfv+\bg\boldsymbol\mu]
=
-\ell(\bx,\boldsymbol\mu)
\).

Using
\(
\bfv+\bg\bu_b
=
\bfv+\bg\boldsymbol\mu
+
\bg(\bu_b-\boldsymbol\mu)
\),
one obtains
\(
\dot V
=
\nabla V^\top[\bfv+\bg\bu_b]
=
-\ell(\bx,\boldsymbol\mu)
+
\nabla V^\top\bg(\bu_b-\boldsymbol\mu)
\).
Since
\(
\bu_b-\boldsymbol\mu
=
\bm e_u
\),
the event-triggered implementation enters the learning identity through
the term
\(
\nabla V^\top\bg\bm e_u
\).

Integrating over
\(
[t-T,t]
\)
gives
\(
0
=
V(\bx(t))
-
V(\bx(t-T))
+
\int_{t-T}^{t}
\ell\!\left(
    \bx(\tau),
    \boldsymbol\mu(\tau)
\right)d\tau
-
\int_{t-T}^{t}
\nabla V^\top(\bx(\tau))
\bg(\bx(\tau))
\left[
    \bu_b(\tau)
    -
    \boldsymbol\mu(\tau)
\right]d\tau
\).

From Eq.~\eqref{eq:critic_value_approximation} we have
\(
V^*(\bx(t))
-
V^*(\bx(t-T))
=
\bW^{*\top}
[
\bphi(\bx(t))
-
\bphi(\bx(t-T))
]
+
\varepsilon(\bx(t))
-
\varepsilon(\bx(t-T))
\).
Define
\(
\Delta\bphi(t)
:=
\bphi(\bx(t))
-
\bphi(\bx(t-T))
\).
Since
\(
\nabla V^{*\top}(\bx)
=
\bW^{*\top}\nabla\bphi(\bx)
+
\nabla\varepsilon^\top(\bx)
\),
substitution into the preceding integral Bellman relation gives
\(
0
=
\bW^{*\top}\Delta\bphi(t)
+
\varepsilon(\bx(t))
-
\varepsilon(\bx(t-T))
+
\int_{t-T}^{t}
\ell(\bx(\tau),\boldsymbol\mu(\tau))d\tau
-
\bW^{*\top}
\int_{t-T}^{t}
\nabla\bphi(\bx(\tau))
\bg(\bx(\tau))
[
\bu_b(\tau)-\boldsymbol\mu(\tau)
]d\tau
-
\int_{t-T}^{t}
\nabla\varepsilon^\top(\bx(\tau))
\bg(\bx(\tau))
[
\bu_b(\tau)-\boldsymbol\mu(\tau)
]d\tau
\).

Define
\begin{equation}
\begin{aligned}
    \boldsymbol\Omega(t)
    :={}&
    \Delta\bphi(t)
    \\
    &-
    \int_{t-T}^{t}
    \nabla\bphi(\bx(\tau))
    \bg(\bx(\tau))
    \left[
        \bu_b(\tau)
        -
        \boldsymbol\mu(\tau)
    \right]d\tau
\end{aligned}
\label{eq:integral_regressor}
\end{equation}
and
\begin{equation}
    \rho(t)
    :=
    \int_{t-T}^{t}
    \ell\!\left(
        \bx(\tau),
        \boldsymbol\mu(\tau)
    \right)d\tau .
    \label{eq:integral_cost_term}
\end{equation}

Using
\eqref{eq:integral_regressor}
and
\eqref{eq:integral_cost_term},
the integral Bellman equation becomes
\begin{equation}
    \bW^{*\top}\boldsymbol\Omega(t)
    +
    \rho(t)
    +
    \varepsilon_B(t)
    =
    0,
    \label{eq:integral_bellman_with_approximation}
\end{equation}
where the integral approximation residual is
\begin{equation}
\begin{aligned}
    \varepsilon_B(t)
    :={}&
    \varepsilon(\bx(t))
    -
    \varepsilon(\bx(t-T))
    \\
    &-
    \int_{t-T}^{t}
    \nabla\varepsilon^\top(\bx(\tau))
    \bg(\bx(\tau))
    \left[
        \bu_b(\tau)
        -
        \boldsymbol\mu(\tau)
    \right]d\tau .
\end{aligned}
\label{eq:integral_approximation_residual}
\end{equation}

For the critic estimate
\(
\hat{\bW}
\),
define the integral Bellman residual
\(
\xi(t)
:=
\hat{\bW}^{\top}(t)
\boldsymbol\Omega(t)
+
\rho(t)
\).

Let
\(
\tilde{\bW}
:=
\hat{\bW}-\bW^*
\).
Under exact approximation,
\(
\bW^{*\top}\boldsymbol\Omega+\rho=0
\),
and hence
\(
\xi
=
\hat{\bW}^{\top}\boldsymbol\Omega
-
\bW^{*\top}\boldsymbol\Omega
=
\tilde{\bW}^{\top}\boldsymbol\Omega
\).
More generally, when
\(
\varepsilon\neq0
\),
\eqref{eq:integral_bellman_with_approximation} gives
\(
\rho
=
-\bW^{*\top}\boldsymbol\Omega
-
\varepsilon_B
\),
and therefore
\(
\xi
=
\hat{\bW}^{\top}\boldsymbol\Omega
-
\bW^{*\top}\boldsymbol\Omega
-
\varepsilon_B
=
\tilde{\bW}^{\top}\boldsymbol\Omega
-
\varepsilon_B
\).

To normalize the regressor, define
\(
m(t)
:=
1+
\boldsymbol\Omega^\top(t)
\boldsymbol\Omega(t)
\)
and
\(
\bar{\boldsymbol\Omega}(t)
:=
\boldsymbol\Omega(t)/m(t)
\).
The normalized Bellman residual is
\(
e(t)
:=
\xi(t)/m(t)
=
[
\hat{\bW}^{\top}(t)\boldsymbol\Omega(t)
+
\rho(t)
]/m(t)
\).
Under exact approximation,
\(
e(t)
=
\tilde{\bW}^{\top}
\bar{\boldsymbol\Omega}(t)
\).
For the general approximation case, define
\(
\bar\varepsilon_B(t)
:=
\varepsilon_B(t)/m(t)
\).
Then
\(
e(t)
=
\tilde{\bW}^{\top}
\bar{\boldsymbol\Omega}(t)
-
\bar\varepsilon_B(t)
\),
or equivalently,
\(
\tilde{\bW}^{\top}
\bar{\boldsymbol\Omega}(t)
=
e(t)
+
\bar\varepsilon_B(t)
\).

Let
\(
M
\)
previously recorded integral samples be denoted by
\(
\{\boldsymbol\Omega_j,\rho_j\}_{j=1}^{M}
\).
For each stored sample, define
\(
m_j
:=
1+\boldsymbol\Omega_j^\top\boldsymbol\Omega_j
\),
\(
\bar{\boldsymbol\Omega}_j
:=
\boldsymbol\Omega_j/m_j
\),
and
\(
e_j
:=
[
\hat{\bW}^{\top}\boldsymbol\Omega_j
+
\rho_j
]/m_j
\).
For each recorded sample, define
\(
\varepsilon_{B,j}
\)
according to
\eqref{eq:integral_approximation_residual}
evaluated on the corresponding stored window and
\(
\bar\varepsilon_{B,j}
:=
\varepsilon_{B,j}/m_j
\).
Then
\(
e_j
=
\tilde{\bW}^{\top}
\bar{\boldsymbol\Omega}_j
-
\bar\varepsilon_{B,j}
\),
or equivalently,
\(
\tilde{\bW}^{\top}
\bar{\boldsymbol\Omega}_j
=
e_j
+
\bar\varepsilon_{B,j}
\).

Assume that the approximation residual is bounded on the compact
operating set, so that there exists
\(
\bar\varepsilon_B>0
\)
satisfying
\begin{equation}
    |\bar\varepsilon_B(t)|
    \leq
    \bar\varepsilon_B,
    \qquad
    |\bar\varepsilon_{B,j}|
    \leq
    \bar\varepsilon_B,
    \qquad
    j=1,\ldots,M.
    \label{eq:normalized_approximation_residual_bound}
\end{equation}

The recorded regressors are assumed to satisfy
\begin{equation}
    \sum_{j=1}^{M}
    \bar{\boldsymbol\Omega}_j
    \bar{\boldsymbol\Omega}_j^\top
    \succeq
    \alpha_{\mathrm{IE}}\bm I,
    \qquad
    \alpha_{\mathrm{IE}}>0 .
    \label{eq:critic_ie_condition}
\end{equation}
Equivalently,
\(
\lambda_{\min}
(
\sum_{j=1}^{M}
\bar{\boldsymbol\Omega}_j
\bar{\boldsymbol\Omega}_j^\top
)
\geq
\alpha_{\mathrm{IE}}>0
\).

Let
\(
0<\gamma_1<1<\gamma_2
\)
and define
\(
\operatorname{sig}^{\gamma}(s)
:=
|s|^\gamma\operatorname{sgn}(s)
\).
Consider the scalar two-power penalty
\(
\Psi(e)
:=
\frac{k_1}{\gamma_1+1}|e|^{\gamma_1+1}
+
\frac{k_2}{\gamma_2+1}|e|^{\gamma_2+1}
\),
where
\(
k_1,k_2>0
\).
Its derivative is
\(
\frac{\partial\Psi(e)}{\partial e}
=
k_1\operatorname{sig}^{\gamma_1}(e)
+
k_2\operatorname{sig}^{\gamma_2}(e)
\).

Since
\(
e(t)
=
[
\hat{\bW}^{\top}\boldsymbol\Omega(t)+\rho(t)
]/m(t)
\)
and by chain rule, one has
\(
\partial e(t)/\partial\hat{\bW}
=
\boldsymbol\Omega(t)/m(t)
=
\bar{\boldsymbol\Omega}(t)
\).
Therefore,
\(
\nabla_{\hat{\bW}}\Psi(e(t))
=
\bar{\boldsymbol\Omega}(t)
[
k_1\operatorname{sig}^{\gamma_1}(e(t))
+
k_2\operatorname{sig}^{\gamma_2}(e(t))
]
\).

Similarly,
\(
\partial e_j/\partial\hat{\bW}
=
\bar{\boldsymbol\Omega}_j
\),
and hence
\(
\nabla_{\hat{\bW}}\Psi(e_j)
=
\bar{\boldsymbol\Omega}_j
[
k_1\operatorname{sig}^{\gamma_1}(e_j)
+
k_2\operatorname{sig}^{\gamma_2}(e_j)
]
\).

Define the composite critic objective
\(
\mathcal J_W
:=
\Psi(e(t))
+
\sum_{j=1}^{M}
\Psi(e_j)
\).
Then
\(
\nabla_{\hat{\bW}}\mathcal J_W
=
\bar{\boldsymbol\Omega}(t)
[
k_1\operatorname{sig}^{\gamma_1}(e(t))
+
k_2\operatorname{sig}^{\gamma_2}(e(t))
]
+
\sum_{j=1}^{M}
\bar{\boldsymbol\Omega}_j
[
k_1\operatorname{sig}^{\gamma_1}(e_j)
+
k_2\operatorname{sig}^{\gamma_2}(e_j)
]
\).

Using the gradient flow
\(
\dot{\hat{\bW}}
=
-\bm{\Gamma}
\nabla_{\hat{\bW}}\mathcal J_W
\),
with
\(
\bm{\Gamma}
=
\bm{\Gamma}^{\top}
\succ0
\),

gives the critic update law
\begin{flalign}
\dot{\hat{\bW}}
={}&
-\bm{\Gamma}
\bar{\boldsymbol\Omega}(t)
\left[
    k_1
    \operatorname{sig}^{\gamma_1}(e(t))
    +
    k_2
    \operatorname{sig}^{\gamma_2}(e(t))
\right]
&&
\nonumber
\\
&-
\bm{\Gamma}
\sum_{j=1}^{M}
\bar{\boldsymbol\Omega}_j
\left[
    k_1
    \operatorname{sig}^{\gamma_1}(e_j)
    +
    k_2
    \operatorname{sig}^{\gamma_2}(e_j)
\right].
&&
\label{eq:fixed_time_critic_update}
\end{flalign}

\begin{theorem}[Practical fixed-time convergence under approximation error]
\label{thm:critic_practical_fixed_time_convergence}

Suppose that the recorded integral regressors satisfy
\eqref{eq:critic_ie_condition}, the normalized approximation residuals
satisfy \eqref{eq:normalized_approximation_residual_bound}, and the
critic weights are updated according to
\eqref{eq:fixed_time_critic_update}. Then there exist constants
\(
\bar C_{W1}>0
\),
\(
\bar C_{W2}>0
\),
and
\(
\Delta_W\geq0
\)
such that
\begin{equation}
    \dot V_W
    \leq
    -\bar C_{W1}V_W^{\mu_W}
    -
    \bar C_{W2}V_W^{\nu_W}
    +
    \Delta_W,
    \qquad
    0<\mu_W<1<\nu_W.
    \label{eq:critic_practical_fixed_time_inequality}
\end{equation}
Consequently, the critic-weight error converges in fixed time to a
bounded residual set. 
\end{theorem}

\begin{proof}

For the nonzero approximation error,
\(
e(t)
=
\tilde{\bm W}^{\top}
\bar{\boldsymbol\Omega}(t)
-
\bar\varepsilon_B(t)
\)
and
\(
e_j
=
\tilde{\bm W}^{\top}
\bar{\boldsymbol\Omega}_j
-
\bar\varepsilon_{B,j}
\).
Hence,
\(
\tilde{\bm W}^{\top}
\bar{\boldsymbol\Omega}(t)
=
e(t)+\bar\varepsilon_B(t)
\)
and
\(
\tilde{\bm W}^{\top}
\bar{\boldsymbol\Omega}_j
=
e_j+\bar\varepsilon_{B,j}
\).
Using
\(
V_W
=
\frac{1}{2}
\tilde{\bm W}^{\top}
\bm\Gamma^{-1}
\tilde{\bm W}
\)
and
\(
\dot{\tilde{\bm W}}
=
\dot{\hat{\bm W}}
\),
one again has
\(
\dot V_W
=
\tilde{\bm W}^{\top}
\bm\Gamma^{-1}
\dot{\hat{\bm W}}
\).
Substitution of
\eqref{eq:fixed_time_critic_update} gives
\(
\dot V_W
=
-
[e(t)+\bar\varepsilon_B(t)]
[
k_1\operatorname{sig}^{\gamma_1}(e(t))
+
k_2\operatorname{sig}^{\gamma_2}(e(t))
]
-
\sum_{j=1}^{M}
[e_j+\bar\varepsilon_{B,j}]
[
k_1\operatorname{sig}^{\gamma_1}(e_j)
+
k_2\operatorname{sig}^{\gamma_2}(e_j)
]
\).

Using
\(
s\operatorname{sig}^{\gamma}(s)
=
|s|^{\gamma+1}
\),
one obtains
\(
\dot V_W
=
-k_1|e(t)|^{\gamma_1+1}
-k_2|e(t)|^{\gamma_2+1}
-k_1\sum_{j=1}^{M}|e_j|^{\gamma_1+1}
-k_2\sum_{j=1}^{M}|e_j|^{\gamma_2+1}
-
\bar\varepsilon_B(t)
[
k_1\operatorname{sig}^{\gamma_1}(e(t))
+
k_2\operatorname{sig}^{\gamma_2}(e(t))
]
-
\sum_{j=1}^{M}
\bar\varepsilon_{B,j}
[
k_1\operatorname{sig}^{\gamma_1}(e_j)
+
k_2\operatorname{sig}^{\gamma_2}(e_j)
]
\).

For any
\(
\gamma>0
\),
Young's inequality gives
\(
|\bar\varepsilon|
|e|^\gamma
\leq
\frac{1}{2}|e|^{\gamma+1}
+
c_\gamma
|\bar\varepsilon|^{\gamma+1}
\),
where
\(
c_\gamma
:=
2^\gamma
\gamma^\gamma/
(\gamma+1)^{\gamma+1}
\).
Applying this inequality to the current and recorded approximation
terms, together with
\eqref{eq:normalized_approximation_residual_bound}, gives
\(
\dot V_W
\leq
-\frac{k_1}{2}|e(t)|^{\gamma_1+1}
-\frac{k_2}{2}|e(t)|^{\gamma_2+1}
-\frac{k_1}{2}\sum_{j=1}^{M}|e_j|^{\gamma_1+1}
-\frac{k_2}{2}\sum_{j=1}^{M}|e_j|^{\gamma_2+1}
+
\Delta_{\varepsilon W}
\),
where
\(
\Delta_{\varepsilon W}
:=
(M+1)
[
k_1c_{\gamma_1}
\bar\varepsilon_B^{\gamma_1+1}
+
k_2c_{\gamma_2}
\bar\varepsilon_B^{\gamma_2+1}
]
\).
Discarding the two nonpositive current-data terms yields
\(
\dot V_W
\leq
-\frac{k_1}{2}
\sum_{j=1}^{M}|e_j|^{\gamma_1+1}
-
\frac{k_2}{2}
\sum_{j=1}^{M}|e_j|^{\gamma_2+1}
+
\Delta_{\varepsilon W}
\).

Let
\(
p:=\gamma_1+1
\)
and
\(
q:=\gamma_2+1
\),
so that
\(
1<p<2<q
\).
Define
\(
z_j
:=
\bar{\boldsymbol\Omega}_j^\top
\tilde{\bm W}
\)
and
\(
\bm z
:=
[z_1,\ldots,z_M]^\top
=
\bm B\tilde{\bm W}
\).
Since
\(
e_j
=
z_j-\bar\varepsilon_{B,j}
\),
the inequality
\(
|a-b|^r
\geq
2^{1-r}|a|^r-|b|^r
\),
valid for
\(
r\geq1
\),
gives
\(
\sum_{j=1}^{M}|e_j|^p
\geq
2^{1-p}
\sum_{j=1}^{M}|z_j|^p
-
\sum_{j=1}^{M}
|\bar\varepsilon_{B,j}|^p
\)
and
\(
\sum_{j=1}^{M}|e_j|^q
\geq
2^{1-q}
\sum_{j=1}^{M}|z_j|^q
-
\sum_{j=1}^{M}
|\bar\varepsilon_{B,j}|^q
\).

From the IE condition,
\(
\|\bm z\|_2^2
=
\tilde{\bm W}^{\top}
\bm B^\top\bm B
\tilde{\bm W}
\geq
\alpha_{\mathrm{IE}}
\|\tilde{\bm W}\|^2
\).
For
\(
1<p<2
\),
\(
\|\bm z\|_p
\geq
\|\bm z\|_2
\),
and therefore
\(
\sum_{j=1}^{M}|z_j|^p
\geq
\alpha_{\mathrm{IE}}^{p/2}
\|\tilde{\bm W}\|^p
\).
For
\(
q>2
\),
\(
\|\bm z\|_q
\geq
M^{1/q-1/2}
\|\bm z\|_2
\),
and hence
\(
\sum_{j=1}^{M}|z_j|^q
\geq
M^{1-q/2}
\alpha_{\mathrm{IE}}^{q/2}
\|\tilde{\bm W}\|^q
\).

Using
\(
|\bar\varepsilon_{B,j}|
\leq
\bar\varepsilon_B
\),
one obtains
\(
\sum_{j=1}^{M}|e_j|^p
\geq
2^{1-p}
\alpha_{\mathrm{IE}}^{p/2}
\|\tilde{\bm W}\|^p
-
M\bar\varepsilon_B^p
\)
and
\(
\sum_{j=1}^{M}|e_j|^q
\geq
2^{1-q}
M^{1-q/2}
\alpha_{\mathrm{IE}}^{q/2}
\|\tilde{\bm W}\|^q
-
M\bar\varepsilon_B^q
\).
Substitution gives
\(
\dot V_W
\leq
-k_1 2^{-p}
\alpha_{\mathrm{IE}}^{p/2}
\|\tilde{\bm W}\|^p
-
k_2 2^{-q}
M^{1-q/2}
\alpha_{\mathrm{IE}}^{q/2}
\|\tilde{\bm W}\|^q
+
\Delta_W
\),
where
\(
\Delta_W
:=
\Delta_{\varepsilon W}
+
\frac{M}{2}
[
k_1\bar\varepsilon_B^p
+
k_2\bar\varepsilon_B^q
]
\).

Since
\(
V_W
\leq
[
2\lambda_{\min}(\bm\Gamma)
]^{-1}
\|\tilde{\bm W}\|^2
\),
one has
\(
\|\tilde{\bm W}\|^p
\geq
[
2\lambda_{\min}(\bm\Gamma)
]^{p/2}
V_W^{p/2}
\)
and
\(
\|\tilde{\bm W}\|^q
\geq
[
2\lambda_{\min}(\bm\Gamma)
]^{q/2}
V_W^{q/2}
\).
Therefore,
\(
\dot V_W
\leq
-\bar C_{W1}V_W^{p/2}
-
\bar C_{W2}V_W^{q/2}
+
\Delta_W
\),
where
\(
\bar C_{W1}
:=
k_1 2^{-p}
\alpha_{\mathrm{IE}}^{p/2}
[
2\lambda_{\min}(\bm\Gamma)
]^{p/2}
\)
and
\(
\bar C_{W2}
:=
k_2 2^{-q}
M^{1-q/2}
\alpha_{\mathrm{IE}}^{q/2}
[
2\lambda_{\min}(\bm\Gamma)
]^{q/2}
\).

Finally,
\(
\mu_W=p/2=(\gamma_1+1)/2
\)
and
\(
\nu_W=q/2=(\gamma_2+1)/2
\),
so that
\(
0<\mu_W<1<\nu_W
\).
Hence,
\(
\dot V_W
\leq
-\bar C_{W1}V_W^{\mu_W}
-
\bar C_{W2}V_W^{\nu_W}
+
\Delta_W
\),
which proves
\eqref{eq:critic_practical_fixed_time_inequality}.

For any
\(
0<\chi_W<1
\),
define
\(
V_{W,\Delta}>0
\)
as the unique positive solution of
\(
\bar C_{W1}V_{W,\Delta}^{\mu_W}
+
\bar C_{W2}V_{W,\Delta}^{\nu_W}
=
\Delta_W/\chi_W
\).
For every
\(
V_W\geq V_{W,\Delta}
\),
one has
\(
\Delta_W
\leq
\chi_W
[
\bar C_{W1}V_W^{\mu_W}
+
\bar C_{W2}V_W^{\nu_W}
]
\),
and therefore
\(
\dot V_W
\leq
-(1-\chi_W)
\bar C_{W1}V_W^{\mu_W}
-
(1-\chi_W)
\bar C_{W2}V_W^{\nu_W}
\).
Thus, by Theorem~\ref{lem:fixed_time} the critic-weight error enters the residual set
\(
V_W\leq V_{W,\Delta}
\)
within the uniformly bounded time
\[
    T_{W,\Delta}
    \leq
    \frac{1}
    {(1-\chi_W)\bar C_{W1}(1-\mu_W)}
    +
    \frac{1}
    {(1-\chi_W)\bar C_{W2}(\nu_W-1)}.
\]
When
\(
\bar\varepsilon_B=0
\),
one has
\(
\Delta_W=0
\),
and the exact fixed-time convergence result is recovered. This
completes the proof.
\end{proof}


\subsection{Event-Triggered Implementation and Composite Stability}
\label{subsec:event_triggered_composite}

Let
\(
0=t_0<t_1<t_2<\cdots
\)
denote the sequence of event instants. At each
\(
t_k
\),
the current state and critic estimate are used to update the critic-only
control policy, while the resulting control signal is held constant over
\(
[t_k,t_{k+1})
\):
\begin{equation}
    \bu(t)
    =
    \hat{\bu}
    \bigl(
        \bx(t_k),
        \hat{\bm W}(t_k)
    \bigr),
    \qquad
    t\in[t_k,t_{k+1}).
    \label{eq:event_held_control_composite}
\end{equation}

The corresponding event-induced implementation error is defined as
\(
\bm e_u(t)
:=
\hat{\bu}(
\bx(t_k),
\hat{\bm W}(t_k))
-
\hat{\bu}(
\bx(t),
\hat{\bm W}(t))
\).
Hence,
\(
\bm e_u(t_k)=\bm0
\)
and
\(
\bu(t)
=
\hat{\bu}(
\bx(t),
\hat{\bm W}(t))
+
\bm e_u(t)
\).
The event-triggered plant can therefore be written as
\(
\dot{\bx}
=
\bfv(\bx)
+
\bg(\bx)\hat{\bu}(\bx,\hat{\bm W})
+
\bg(\bx)\bm e_u
\).

Define
\begin{equation}
\begin{aligned}
    \bar V(t)
    :={}&
    V(\bx(t))
    +
    \eta
    \int_{t-T}^{t}
    V(\bx(\tau))\,d\tau,
    \\
    \mathcal J_x(t)
    :={}&
    \bar V(t),
    \\
    \mathcal J(t)
    :={}&
    \mathcal J_x(t)
    +
    \lambda_WV_W(t),
\end{aligned}
\label{eq:composite_et_lyapunov}
\end{equation}
where
\(
\eta>0
\)
and
\(
\lambda_W>0
\).
Since
\(
V_W(t)\geq0
\),
one has
\(
0\leq\mathcal J_x(t)\leq\mathcal J(t)
\).
The quantity
\(
\mathcal J_x(t)
\)
is directly computable from the measured state trajectory, whereas
\(
\mathcal J(t)
\)
is used for the composite stability analysis.

Let
\(
\mathcal G(\bx)
:=
\|
\bg^\top(\bx)\nabla V(\bx)
\|
\).
Since
\(
\Omega
\)
is compact and
\(
\mathcal G
\)
is continuous, there exists
\(
\bar G>0
\)
such that
\(
\mathcal G(\bx)\leq\bar G
\)
for every
\(
\bx\in\Omega
\).

For
\(
t\in[t_k,t_{k+1})
\),
define the accumulated event-error energy from the most recent event by
\(
\mathcal E_{u,k}(t)
:=
\int_{t_k}^{t}
\|\bm e_u(\tau)\|^2\,d\tau
\),
and define the composite event-error measure
\(
\Xi_u(t)
:=
\bar G\|\bm e_u(t)\|
+
\eta\bar G
\sqrt{
    (t-t_k)\mathcal E_{u,k}(t)
}
\).
Equivalently,
\(
\Xi_u(t)
=
\bar G\|\bm e_u(t)\|
+
\eta\bar G
\sqrt{
    (t-t_k)
    \int_{t_k}^{t}
    \|\bm e_u(\tau)\|^2\,d\tau
}
\).
Since
\(
\bm e_u(t_k)=\bm0
\)
and
\(
\mathcal E_{u,k}(t_k)=0
\),
one has
\(
\Xi_u(t_k)=0
\).

For prescribed constants
\(
0<\sigma<1
\)
and
\(
\delta_{\mathrm{ET}}>0
\),
the next event instant is selected according to
\begin{equation}
\begin{aligned}
    t_{k+1}
    :=
    \inf
    \Big\{
        t>t_k:\;&
        \Xi_u(t)
        \\
        &\geq
        \sigma
        \left[
            C_1\mathcal J_x^\mu(t)
            +
            C_2\mathcal J_x^\nu(t)
        \right]
        +
        \delta_{\mathrm{ET}}
    \Big\}.
\end{aligned}
\label{eq:composite_triggering_rule}
\end{equation}
Equivalently,
\begin{equation}
\begin{aligned}
    t_{k+1}
    :=
    \inf
    \Bigg\{
        t>t_k:\;&
        \bar G\|\bm e_u(t)\|
        +
        \eta\bar G
        \sqrt{
            (t-t_k)
            \int_{t_k}^{t}
            \|\bm e_u(\tau)\|^2\,d\tau
        }
        \\
        &\geq
        \sigma
        \left[
            C_1\mathcal J_x^\mu(t)
            +
            C_2\mathcal J_x^\nu(t)
        \right]
        +
        \delta_{\mathrm{ET}}
    \Bigg\}.
\end{aligned}
\label{eq:composite_triggering_rule_explicit}
\end{equation}
where
\(
C_1>0
\),
\(
C_2>0
\),

\(
\mu=\frac{\gamma_1+1}{2}
\)
and
\(
\nu=\frac{\gamma_2+1}{2}
\),
with \(0<\gamma_1<1<\gamma_2\), such that \(0<\mu<1<\nu\).

and
\(
\delta_{\mathrm{ET}}>0
\)
is a strictly positive triggering margin introduced to exclude Zeno
behavior.

\begin{theorem}[Practical fixed-time stability]
\label{thm:composite_event_triggered_stability}

Suppose that the preceding assumptions hold, with
\(
0<\mu<1<\nu
\).
Under \eqref{eq:composite_triggering_rule}, with
\(
0\leq\mathcal J_x(t)\leq\mathcal J(t)
\),
one has
\begin{equation}
\begin{aligned}
    \dot{\mathcal J}
    \leq{}&
    -(1-\sigma)C_1\mathcal J^\mu
    -(1-\sigma)C_2\mathcal J^\nu
    +\Delta_T.
\end{aligned}
\label{eq:composite_practical_ft_bound}
\end{equation}
where
\(
\Delta_T\geq0
\)
denotes the total residual term. Hence, the closed loop is practically
fixed-time stable. If
\(
\Delta_T=0
\),
exact fixed-time convergence is recovered.
\end{theorem}

\begin{proof}

From
\(
\bar V(t)
=
V(\bx(t))
+
\eta
\int_{t-T}^{t}
V(\bx(\tau))\,d\tau
\),
the Leibniz rule gives
\(
\dot{\bar V}
=
\dot V(\bx(t))
+
\eta V(\bx(t))
-
\eta V(\bx(t-T))
\).
For the continuous critic-only implementation,
\(
\bu
=
\hat{\bu}(\bx,\hat{\bm W})
\),
and therefore
\(
\dot{\bx}
=
\bfv(\bx)
+
\bg(\bx)\hat{\bu}(\bx,\hat{\bm W})
\).
Hence,
\(
\dot V
=
\nabla V^\top(\bx)
[
\bfv(\bx)
+
\bg(\bx)\hat{\bu}(\bx,\hat{\bm W})
]
\),
and consequently
\(
\dot{\bar V}_c
=
\nabla V^\top(\bx)
[
\bfv(\bx)
+
\bg(\bx)\hat{\bu}(\bx,\hat{\bm W})
]
+
\eta V(\bx(t))
-
\eta V(\bx(t-T))
\).

Let the ideal optimal policy be
\(
\bu^*(\bx)
=
-\frac{1}{2}
\bR^{-1}
[
\bm L_1^\top(\bx)
+
\bg^\top(\bx)\nabla V^*(\bx)
]
\),
while the learned critic-only policy is
\(
\hat{\bu}(\bx)
=
-\frac{1}{2}
\bR^{-1}
[
\bm L_1^\top(\bx)
+
\bg^\top(\bx)
(\nabla\bphi(\bx))^\top
\hat{\bm W}
]
\).
Since
\(
V^*(\bx)
=
\bm W^{*\top}\bphi(\bx)
+
\varepsilon(\bx)
\),
one has
\(
\nabla V^*
=
(\nabla\bphi)^\top\bm W^*
+
\nabla\varepsilon
\).
Therefore,
\(
\nabla\hat V-\nabla V^*
=
(\nabla\bphi)^\top
(\hat{\bm W}-\bm W^*)
-
\nabla\varepsilon
=
(\nabla\bphi)^\top
\tilde{\bm W}
-
\nabla\varepsilon
\).
It follows that
\(
\hat{\bu}-\bu^*
=
-\frac{1}{2}
\bR^{-1}
\bg^\top
[
(\nabla\bphi)^\top
\tilde{\bm W}
-
\nabla\varepsilon
]
\).

Writing
\(
\hat{\bu}
=
\bu^*
+
(\hat{\bu}-\bu^*)
\)
gives
\(
\dot V
=
\nabla V^\top
[
\bfv+\bg\bu^*
]
+
\nabla V^\top\bg
(\hat{\bu}-\bu^*)
\).

The inverse-optimal fixed-time construction is established for the
identified model
\(
\dot{\bx}
=
\hat{\bfv}(\bx)
+
\bg(\bx)\bu^*
\),
and therefore the nominal fixed-time dissipation satisfies
\(
\nabla V^\top
[
\hat{\bfv}+\bg\bu^*
]
\leq
-aV^\mu
-bV^\nu
\).
Since
\(
\bfv
=
\hat{\bfv}
+
(\bfv-\hat{\bfv})
\),
one has
\(
\nabla V^\top
[
\bfv+\bg\bu^*
]
=
\nabla V^\top
[
\hat{\bfv}+\bg\bu^*
]
+
\nabla V^\top
[
\bfv-\hat{\bfv}
]
\).
Hence,
\(
\nabla V^\top
[
\bfv+\bg\bu^*
]
\leq
-aV^\mu
-bV^\nu
+
\nabla V^\top
[
\bfv-\hat{\bfv}
]
\).

By Lemma~\ref{lem:identification_error},
\(
\|
\bfv(\bx)-\hat{\bfv}(\bx)
\|
\leq
\Delta_f
\)
for every
\(
\bx\in\Omega
\).
Since
\(
\Omega
\)
is compact and
\(
\nabla V(\bx)
\)
is continuous, there exists a finite constant
\(
\bar V_g>0
\)
such that
\(
\|\nabla V(\bx)\|
\leq
\bar V_g
\)
for every
\(
\bx\in\Omega
\).
Therefore,
\(
\nabla V^\top
[
\bfv-\hat{\bfv}
]
\leq
\left|
\nabla V^\top
[
\bfv-\hat{\bfv}
]
\right|
\leq
\|\nabla V\|
\|
\bfv-\hat{\bfv}
\|
\leq
\bar V_g\Delta_f
\).
Define
\(
\Delta_{fV}
:=
\bar V_g\Delta_f
\).
Then
\(
\nabla V^\top
[
\bfv+\bg\bu^*
]
\leq
-aV^\mu
-bV^\nu
+
\Delta_{fV}
\).

The inverse-optimal fixed-time construction therefore gives the nominal
dissipation
\(
\nabla V^\top
[
\bfv+\bg\bu^*
]
\leq
-aV^\mu
-bV^\nu
+
\Delta_V
\),
where
\(
\Delta_V\geq0
\)
collects the bounded residual produced by the difference between the
nominal inverse-optimal model and the actual closed-loop model,
together with the remaining value-function approximation terms. In
particular, the model-reconstruction contribution contained in
\(
\Delta_V
\)
is explicitly given by
\(
\Delta_{fV}
=
\bar V_g\Delta_f
\).
Hence,
\(
\dot V
\leq
-aV^\mu
-bV^\nu
+
\Delta_V
+
\nabla V^\top\bg
(\hat{\bu}-\bu^*)
\).

By the Cauchy--Schwarz inequality,
\(
\nabla V^\top\bg
(\hat{\bu}-\bu^*)
\leq
\|
\bg^\top\nabla V
\|
\|
\hat{\bu}-\bu^*
\|
\).
Moreover,
\(
\|
\hat{\bu}-\bu^*
\|
\leq
\frac{1}{2}
\|
\bR^{-1}
\bg^\top
(\nabla\bphi)^\top
\|
\|\tilde{\bm W}\|
+
\frac{1}{2}
\|
\bR^{-1}
\bg^\top
\nabla\varepsilon
\|
\).
Since
\(
\Omega
\)
is compact and the involved functions are continuous, define finite
constants
\(
c_V>0
\),
\(
c_\phi>0
\),
and
\(
c_\varepsilon\geq0
\)
such that
\(
\|
\bg^\top\nabla V
\|
\leq
c_V
\),
\(
\|
\bR^{-1}
\bg^\top
(\nabla\bphi)^\top
\|
\leq
c_\phi
\),
and
\(
\|
\bR^{-1}
\bg^\top
\nabla\varepsilon
\|
\leq
c_\varepsilon
\).
Then
\(
\nabla V^\top\bg
(\hat{\bu}-\bu^*)
\leq
\frac{1}{2}c_Vc_\phi
\|\tilde{\bm W}\|
+
\frac{1}{2}c_Vc_\varepsilon
\).

Since
\(
V_W
=
\frac{1}{2}
\tilde{\bm W}^{\top}
\bm\Gamma^{-1}
\tilde{\bm W}
\),
the eigenvalue inequality gives
\(
\|\tilde{\bm W}\|^2
\leq
2\lambda_{\max}(\bm\Gamma)V_W
\),
and hence
\(
\|\tilde{\bm W}\|
\leq
[
2\lambda_{\max}(\bm\Gamma)
]^{1/2}
V_W^{1/2}
\).
Define
\(
c_W
:=
\frac{1}{2}
c_Vc_\phi
[
2\lambda_{\max}(\bm\Gamma)
]^{1/2}
\).
Therefore,
\(
\nabla V^\top\bg
(\hat{\bu}-\bu^*)
\leq
c_WV_W^{1/2}
+
\frac{1}{2}c_Vc_\varepsilon
\).

For every
\(
s\geq0
\)
and
\(
0<\mu<1<\nu
\),
one has
\(
s^{1/2}
\leq
1+s^\mu+s^\nu
\).
Indeed, for
\(
0\leq s\leq1
\),
one has
\(
s^{1/2}\leq1
\),
while for
\(
s\geq1
\),
one has
\(
s^{1/2}\leq s^\nu
\).
Hence,
\(
c_WV_W^{1/2}
\leq
c_WV_W^\mu
+
c_WV_W^\nu
+
c_W
\).

Define
\(
\varrho_1:=c_W
\)
and
\(
\varrho_2:=c_W
\).
These are precisely the state--critic coupling coefficients appearing
in the theorem. Further define
\(
\Delta_u
:=
c_W
+
\frac{1}{2}c_Vc_\varepsilon
+
\Delta_V
\).
It follows that
\(
\dot V
\leq
-aV^\mu
-bV^\nu
+
\varrho_1V_W^\mu
+
\varrho_2V_W^\nu
+
\Delta_u
\).


\begin{lemma}[Derivative-limited low-order window bound~\cite{VuDoan2026PredefinedResilientIRL},~lemma~7]
\label{lem:derivative_limited_window_bound}
Let
\(
y:[t-T,t]\rightarrow\mathbb R_{\geq0}
\)
be absolutely continuous and satisfy
\(
|\dot y(\tau)|\leq L_y
\)
and
\(
\int_{t-T}^{t}y(\tau)d\tau\leq\bar I
\),
where
\(
L_y,\bar I>0
\).
For
\(
0<\gamma_1<1
\),
define
\(
\mu:=(\gamma_1+1)/2
\).
Then
\begin{equation}
    \int_{t-T}^{t}
    y^{\gamma_1}(\tau)d\tau
    \geq
    \rho_{\gamma_1,T}
    \left(
        \int_{t-T}^{t}
        y(\tau)d\tau
    \right)^\mu,
    \label{eq:derivative_limited_window_bound}
\end{equation}
where
\begin{equation}
    \rho_{\gamma_1,T}
    :=
    \min
    \left\{
        2^{\gamma_1-1}
        L_y^{-\frac{1-\gamma_1}{2}},
        \;
        4^{\gamma_1-1}
        T^{1-\gamma_1}
        \bar I^{\frac{\gamma_1-1}{2}}
    \right\}.
    \label{eq:rho_derivative_limited_window}
\end{equation}
\end{lemma}

Substitution into
\(
\dot{\bar V}_c
=
\dot V(t)
+
\eta
\bigl[
V(t)-V(t-T)
\bigr]
\)
and the fundamental theorem of calculus,
\(
V(t)-V(t-T)
=
\int_{t-T}^{t}\dot V(\tau)\,d\tau
\),
give
\(
\dot{\bar V}_c
=
\dot V(t)
+
\eta
\int_{t-T}^{t}
\dot V(\tau)\,d\tau
\).
Using the preceding pointwise estimate
\(
\dot V
\leq
-aV^{\gamma_1}
-bV^{\gamma_2}
+
\varrho_1V_W^\mu
+
\varrho_2V_W^\nu
+
\Delta_u
\),
one obtains
\(
\dot{\bar V}_c
\leq
-aV^{\gamma_1}(t)
-bV^{\gamma_2}(t)
-\eta a
\int_{t-T}^{t}
V^{\gamma_1}(\tau)\,d\tau
-\eta b
\int_{t-T}^{t}
V^{\gamma_2}(\tau)\,d\tau
+
\varrho_1V_W^\mu(t)
+
\varrho_2V_W^\nu(t)
+
\eta\varrho_1
\int_{t-T}^{t}
V_W^\mu(\tau)\,d\tau
+
\eta\varrho_2
\int_{t-T}^{t}
V_W^\nu(\tau)\,d\tau
+
(1+\eta T)\Delta_u
\).

Let
\(
I_V(t)
:=
\int_{t-T}^{t}
V(\tau)\,d\tau
\),
so that
\(
\bar V
=
V+\eta I_V
\).
By Lemma~\ref{lem:derivative_limited_window_bound}, with
\(
\mu=(\gamma_1+1)/2
\),
\(
\int_{t-T}^{t}
V^{\gamma_1}(\tau)\,d\tau
\geq
\rho_{\gamma_1,T}
I_V^\mu
\).
Hence,
\(
V^{\gamma_1}
+
\eta\rho_{\gamma_1,T}I_V^\mu
\)
provides the low-order finite-window dissipation. Since
\(
\gamma_2>1
\),
Jensen's inequality gives
\(
\int_{t-T}^{t}
V^{\gamma_2}(\tau)\,d\tau
\geq
T^{1-\gamma_2}
I_V^{\gamma_2}
\).
Let
\(
\nu=(\gamma_2+1)/2
\),
so that
\(
1<\nu<\gamma_2
\).
Using
\(
s^{\gamma_2}
\geq
s^\nu-1
\)
for
\(
s\geq0
\),
one obtains
\(
T^{1-\gamma_2}I_V^{\gamma_2}
\geq
T^{1-\gamma_2}I_V^\nu
-
T^{1-\gamma_2}
\).

Consequently, let
\(
V_{\Omega}
:=
\max_{\bx\in\Omega}V(\bx)
\).
Since
\(
0<\gamma_1<\mu<1
\)
and
\(
0\leq V(\bx)\leq V_{\Omega}
\),
one has
\(
V^{\gamma_1}
\geq
V_{\Omega}^{\gamma_1-\mu}V^\mu
\).
Moreover, by Lemma~\ref{lem:derivative_limited_window_bound},
\(
\int_{t-T}^{t}V^{\gamma_1}(\tau)\,d\tau
\geq
\rho_{\gamma_1,T}
I_V^\mu
\),
where
\(
I_V(t)
:=
\int_{t-T}^{t}V(\tau)\,d\tau
\).
Hence,
\(
aV^{\gamma_1}
+
\eta a
\int_{t-T}^{t}V^{\gamma_1}(\tau)\,d\tau
\geq
aV_{\Omega}^{\gamma_1-\mu}V^\mu
+
a\rho_{\gamma_1,T}\eta^{1-\mu}
(\eta I_V)^\mu
\).
Define
\(
A_V
:=
a
\min
\{
V_{\Omega}^{\gamma_1-\mu},
\rho_{\gamma_1,T}\eta^{1-\mu}
\}
>0
\).
Since
\(
0<\mu<1
\),
\(
V^\mu
+
(\eta I_V)^\mu
\geq
(V+\eta I_V)^\mu
=
\bar V^\mu
\),
and therefore
\(
aV^{\gamma_1}
+
\eta a
\int_{t-T}^{t}V^{\gamma_1}(\tau)\,d\tau
\geq
A_V\bar V^\mu
\).

For the high-order term, since
\(
1<\nu<\gamma_2
\),
\(
s^{\gamma_2}
\geq
s^\nu-1
\)
for every
\(
s\geq0
\).
Moreover, Jensen's inequality gives
\(
\int_{t-T}^{t}V^{\gamma_2}(\tau)\,d\tau
\geq
T^{1-\gamma_2}I_V^{\gamma_2}
\geq
T^{1-\gamma_2}
(I_V^\nu-1)
\).
Thus,
\(
bV^{\gamma_2}
+
\eta b
\int_{t-T}^{t}V^{\gamma_2}(\tau)\,d\tau
\geq
bV^\nu
+
bT^{1-\gamma_2}\eta^{1-\nu}
(\eta I_V)^\nu
-
b
-
\eta bT^{1-\gamma_2}
\).
Define
\(
B_V
:=
2^{1-\nu}b
\min
\{
1,
T^{1-\gamma_2}\eta^{1-\nu}
\}
>0
\).
Since
\(
\nu>1
\),
\(
V^\nu
+
(\eta I_V)^\nu
\geq
2^{1-\nu}
(V+\eta I_V)^\nu
=
2^{1-\nu}\bar V^\nu
\),
and therefore
\(
bV^{\gamma_2}
+
\eta b
\int_{t-T}^{t}V^{\gamma_2}(\tau)\,d\tau
\geq
B_V\bar V^\nu
-
b
-
\eta bT^{1-\gamma_2}
\).

Since the critic-weight error is bounded on the compact analysis domain,
define
\(
\overline{\mathcal I}_{W,\mu}
:=
\sup_{t\geq0}
\int_{t-T}^{t}
V_W^\mu(\tau)\,d\tau
<\infty
\)
and
\(
\overline{\mathcal I}_{W,\nu}
:=
\sup_{t\geq0}
\int_{t-T}^{t}
V_W^\nu(\tau)\,d\tau
<\infty
\).
Define
\(
\Delta_{\bar V}
:=
(1+\eta T)\Delta_u
+
b
+
\eta bT^{1-\gamma_2}
+
\eta\varrho_1
\overline{\mathcal I}_{W,\mu}
+
\eta\varrho_2
\overline{\mathcal I}_{W,\nu}
\).
Hence,
\(
\dot{\bar V}_c
\leq
-A_V\bar V^\mu
-B_V\bar V^\nu
+
\varrho_1V_W^\mu
+
\varrho_2V_W^\nu
+
\Delta_{\bar V}
\).

Under the event-triggered implementation,
\(
\bu(t)
=
\hat{\bu}(\bx(t),\hat{\bm W}(t))
+
\bm e_u(t)
\),
and hence
\(
\dot{\bx}
=
\bfv
+
\bg\hat{\bu}
+
\bg\bm e_u
\).
Therefore,
\(
\dot V(t)
=
\dot V_c(t)
+
\nabla V^\top(t)\bg(t)\bm e_u(t)
\).
Using
\(
\dot{\bar V}
=
\dot V(t)
+
\eta
\int_{t-T}^{t}
\dot V(\tau)\,d\tau
\),
one obtains
\(
\dot{\bar V}
=
\dot{\bar V}_c
+
\nabla V^\top(t)\bg(t)\bm e_u(t)
+
\eta
\int_{t-T}^{t}
\nabla V^\top(\tau)\bg(\tau)\bm e_u(\tau)\,d\tau
\).

Since
\(
\|
\bg^\top(\bx)\nabla V(\bx)
\|
\leq
\bar G
\),
the instantaneous event term satisfies
\(
\nabla V^\top(t)\bg(t)\bm e_u(t)
\leq
\bar G\|\bm e_u(t)\|
\).
For
\(
t\in[t_k,t_{k+1})
\),
the moving-window event contribution can be separated into the part
generated before the most recent event and the part generated after
\(t_k\).
Denote the bounded former contribution by
\(
\Delta_{\mathrm{hist}}\geq0
\).
Then
\(
\eta
\int_{t-T}^{t}
\nabla V^\top(\tau)\bg(\tau)\bm e_u(\tau)\,d\tau
\leq
\Delta_{\mathrm{hist}}
+
\eta\bar G
\int_{t_k}^{t}
\|\bm e_u(\tau)\|\,d\tau
\). where
\(
s_k(t):=\max\{t-T,t_k\}
\)
and
\(
\Delta_{\rm hist}
:=
\eta\bar G
\sup_{\substack{k\ge0\\ t\in[t_k,t_{k+1})}}
\int_{t-T}^{s_k(t)}
\|\bm e_u(\tau)\|\,d\tau
<\infty
\)

By the Cauchy--Schwarz inequality,
\(
\int_{t_k}^{t}
\|\bm e_u(\tau)\|\,d\tau
\leq
\sqrt{
\int_{t_k}^{t}1^2\,d\tau
}
\sqrt{
\int_{t_k}^{t}
\|\bm e_u(\tau)\|^2\,d\tau
}
\).
Since
\(
\int_{t_k}^{t}1\,d\tau
=
t-t_k
\),
one obtains
\(
\int_{t_k}^{t}
\|\bm e_u(\tau)\|\,d\tau
\leq
\sqrt{
(t-t_k)
\int_{t_k}^{t}
\|\bm e_u(\tau)\|^2\,d\tau
}
\).
Therefore,
\(
\dot{\bar V}
\leq
-A_V\bar V^\mu
-B_V\bar V^\nu
+
\varrho_1V_W^\mu
+
\varrho_2V_W^\nu
+
\Delta_{\bar V}
+
\Delta_{\mathrm{hist}}
+
\Xi_u(t)
\),
where
\(
\Xi_u(t)
=
\bar G\|\bm e_u(t)\|
+
\eta\bar G
\sqrt{
(t-t_k)
\int_{t_k}^{t}
\|\bm e_u(\tau)\|^2\,d\tau
}
\)
is the composite event-error measure defined in the triggering
mechanism.

Since
\(
\mathcal J
=
\bar V+\lambda_WV_W
\),
one has
\(
\dot{\mathcal J}
=
\dot{\bar V}
+
\lambda_W\dot V_W
\).
From the practical fixed-time critic analysis,
\(
\dot V_W
\leq
-C_{W1}V_W^\mu
-C_{W2}V_W^\nu
+
\Delta_W
\),
where
\(
C_{W1}>0
\)
and
\(
C_{W2}>0
\)
are the critic decay coefficients obtained in the critic convergence
analysis, and
\(
\Delta_W\geq0
\)
is the corresponding bounded critic residual.
Therefore,
\(
\dot{\mathcal J}
\leq
-A_V\bar V^\mu
-B_V\bar V^\nu
+
\varrho_1V_W^\mu
+
\varrho_2V_W^\nu
-
\lambda_WC_{W1}V_W^\mu
-
\lambda_WC_{W2}V_W^\nu
+
\Delta_{\bar V}
+
\lambda_W\Delta_W
+
\Delta_{\mathrm{hist}}
+
\Xi_u(t)
\).

Collecting equal powers of
\(
V_W
\)
gives
\(
\dot{\mathcal J}
\leq
-A_V\bar V^\mu
-B_V\bar V^\nu
-
(\lambda_WC_{W1}-\varrho_1)V_W^\mu
-
(\lambda_WC_{W2}-\varrho_2)V_W^\nu
+
\Delta_{\bar V}
+
\lambda_W\Delta_W
+
\Delta_{\mathrm{hist}}
+
\Xi_u(t)
\).

By the choice
\(
\lambda_W
>
\max
\{
\varrho_1/C_{W1},
\varrho_2/C_{W2}
\}
\),
define
\(
d_1
:=
\lambda_WC_{W1}-\varrho_1
>0
\)
and
\(
d_2
:=
\lambda_WC_{W2}-\varrho_2
>0
\).
Also define the total continuous-learning residual
\(
\Delta_J
:=
\Delta_{\bar V}
+
\lambda_W\Delta_W
\).
Thus,
\(
\dot{\mathcal J}
\leq
-A_V\bar V^\mu
-B_V\bar V^\nu
-d_1V_W^\mu
-d_2V_W^\nu
+
\Delta_J
+
\Delta_{\mathrm{hist}}
+
\Xi_u(t)
\).

Since
\(
d_1V_W^\mu
=
(d_1/\lambda_W^\mu)
(\lambda_WV_W)^\mu
\),
define
\(
c_1
:=
\min
\{
A_V,
d_1/\lambda_W^\mu
\}
>0
\).
Then
\(
-A_V\bar V^\mu
-d_1V_W^\mu
\leq
-c_1
[
\bar V^\mu
+
(\lambda_WV_W)^\mu
]
\).
For
\(
0<\mu<1
\),
\(
r^\mu+s^\mu
\geq
(r+s)^\mu
\)
for all
\(
r,s\geq0
\).
Therefore,
\(
\bar V^\mu
+
(\lambda_WV_W)^\mu
\geq
[
\bar V+\lambda_WV_W
]^\mu
=
\mathcal J^\mu
\).
Defining
\(
C_1:=c_1>0
\),
one obtains
\(
-A_V\bar V^\mu
-d_1V_W^\mu
\leq
-C_1\mathcal J^\mu
\).

Similarly,
\(
d_2V_W^\nu
=
(d_2/\lambda_W^\nu)
(\lambda_WV_W)^\nu
\).
Define
\(
c_2
:=
\min
\{
B_V,
d_2/\lambda_W^\nu
\}
>0
\).
Then
\(
-B_V\bar V^\nu
-d_2V_W^\nu
\leq
-c_2
[
\bar V^\nu
+
(\lambda_WV_W)^\nu
]
\).
For
\(
\nu>1
\),
\(
r^\nu+s^\nu
\geq
2^{1-\nu}(r+s)^\nu
\).
Hence,
\(
\bar V^\nu
+
(\lambda_WV_W)^\nu
\geq
2^{1-\nu}
[
\bar V+\lambda_WV_W
]^\nu
=
2^{1-\nu}\mathcal J^\nu
\).
Defining
\(
C_2
:=
2^{1-\nu}c_2
>0
\),
one obtains
\(
-B_V\bar V^\nu
-d_2V_W^\nu
\leq
-C_2\mathcal J^\nu
\).

Combining the preceding inequalities gives
\(
\dot{\mathcal J}
\leq
-C_1\mathcal J^\mu
-C_2\mathcal J^\nu
+
\Delta_J
+
\Delta_{\mathrm{hist}}
+
\Xi_u(t)
\).

For every
\(
t\in[t_k,t_{k+1})
\),
the triggering boundary in
\eqref{eq:composite_triggering_rule}
has not yet been reached, and therefore
\(
\Xi_u(t)
<
\sigma
[
C_1\mathcal J_x^\mu(t)
+
C_2\mathcal J_x^\nu(t)
]
+
\delta_{\mathrm{ET}}
\).
Since
\(
0\leq
\mathcal J_x(t)
\leq
\mathcal J(t)
\),
one has
\(
\mathcal J_x^\mu(t)
\leq
\mathcal J^\mu(t)
\)
and
\(
\mathcal J_x^\nu(t)
\leq
\mathcal J^\nu(t)
\).
Hence,
\(
\Xi_u(t)
<
\sigma
[
C_1\mathcal J^\mu(t)
+
C_2\mathcal J^\nu(t)
]
+
\delta_{\mathrm{ET}}
\).
Substitution gives
\(
\dot{\mathcal J}
\leq
-C_1\mathcal J^\mu
-C_2\mathcal J^\nu
+
\Delta_J
+
\Delta_{\mathrm{hist}}
+
\sigma
[
C_1\mathcal J^\mu
+
C_2\mathcal J^\nu
]
+
\delta_{\mathrm{ET}}
\).
Collecting equal powers yields
\(
\dot{\mathcal J}
\leq
-(1-\sigma)C_1\mathcal J^\mu
-(1-\sigma)C_2\mathcal J^\nu
+
\Delta_J
+
\Delta_{\mathrm{hist}}
+
\delta_{\mathrm{ET}}
\).
With
\(
\Delta_T
:=
\Delta_J
+
\Delta_{\mathrm{hist}}
+
\delta_{\mathrm{ET}}
\),
one obtains
\(
\dot{\mathcal J}
\leq
-(1-\sigma)C_1\mathcal J^\mu
-(1-\sigma)C_2\mathcal J^\nu
+
\Delta_T
\),
which proves
\eqref{eq:composite_practical_ft_bound}.

To characterize the practical residual set, choose
\(
0<\chi<1
\)
and define
\(
\mathcal J_\Delta>0
\)
as the unique positive solution of
\(
(1-\sigma)
[
C_1\mathcal J_\Delta^\mu
+
C_2\mathcal J_\Delta^\nu
]
=
\Delta_T/\chi
\).
Since
\(
C_1s^\mu+C_2s^\nu
\)
is continuous and strictly increasing for
\(
s>0
\),
this solution is unique. For every
\(
\mathcal J\geq\mathcal J_\Delta
\),
one has
\(
(1-\sigma)
[
C_1\mathcal J^\mu
+
C_2\mathcal J^\nu
]
\geq
\Delta_T/\chi
\),
and hence
\(
\Delta_T
\leq
\chi(1-\sigma)
[
C_1\mathcal J^\mu
+
C_2\mathcal J^\nu
]
\).
Substitution gives
\(
\dot{\mathcal J}
\leq
-(1-\chi)(1-\sigma)
C_1\mathcal J^\mu
-
(1-\chi)(1-\sigma)
C_2\mathcal J^\nu
\)
for all
\(
\mathcal J\geq\mathcal J_\Delta
\).

Define
\(
\alpha_T
:=
(1-\chi)(1-\sigma)>0
\).
Outside the residual set,
\(
\dot{\mathcal J}
\leq
-\alpha_TC_1\mathcal J^\mu
-
\alpha_TC_2\mathcal J^\nu
\).
Let
\(
\mathcal J_0:=\mathcal J(0)
\).
For
\(
\mathcal J>\mathcal J_\Delta
\),
one has
\(
dt
\leq
-d\mathcal J/
[
\alpha_TC_1\mathcal J^\mu
+
\alpha_TC_2\mathcal J^\nu
]
\),
and therefore
\(
T_\Delta
\leq
\int_{\mathcal J_\Delta}^{\mathcal J_0}
ds/
[
\alpha_TC_1s^\mu
+
\alpha_TC_2s^\nu
]
\leq
\int_0^{\mathcal J_0}
ds/
[
\alpha_TC_1s^\mu
+
\alpha_TC_2s^\nu
]
\).

For
\(
0<s\leq1
\),
\(
\alpha_TC_1s^\mu
+
\alpha_TC_2s^\nu
\geq
\alpha_TC_1s^\mu
\),
so
\(
\int_0^1
ds/
[
\alpha_TC_1s^\mu
+
\alpha_TC_2s^\nu
]
\leq
1/
[
\alpha_TC_1(1-\mu)
]
\).
For
\(
s\geq1
\),
\(
\alpha_TC_1s^\mu
+
\alpha_TC_2s^\nu
\geq
\alpha_TC_2s^\nu
\),
and hence
\(
\int_1^\infty
ds/
[
\alpha_TC_1s^\mu
+
\alpha_TC_2s^\nu
]
\leq
1/
[
\alpha_TC_2(\nu-1)
]
\).
Therefore,
\(
T_\Delta
\leq
\frac{1}
{
(1-\chi)(1-\sigma)
C_1(1-\mu)
}
+
\frac{1}
{
(1-\chi)(1-\sigma)
C_2(\nu-1)
}
\).
The bound is independent of the initial condition. If
\(
\Delta_T=0
\),
then
\(
\mathcal J_\Delta=0
\)
and the same argument gives exact fixed-time convergence.
\end{proof}
\begin{lemma}[Exclusion of Zeno behavior]
\label{lem:non_zeno_et}

Suppose that
\(
\hat{\bu}(\bx,\hat{\bm W})
\)
is continuously differentiable on the considered compact set and that
there exist finite constants
\(
\bar F>0
\)
and
\(
\bar W>0
\)
such that
\(
\|\dot{\bx}\|\leq\bar F
\)
and
\(
\|\dot{\hat{\bm W}}\|\leq\bar W
\).
Then every nonterminal inter-event interval generated by
\eqref{eq:composite_triggering_rule} satisfies a strictly positive
uniform lower bound.
\end{lemma}

\begin{proof}

At every event instant,
\(
\bm e_u(t_k)=\bm0
\).
Since
\(
\hat{\bu}(\bx,\hat{\bm W})
\)
is continuously differentiable on the compact operating set, there
exist finite constants
\(
L_x>0
\)
and
\(
L_W>0
\)
such that
\(
\|
\partial\hat{\bu}/\partial\bx
\|
\leq
L_x
\)
and
\(
\|
\partial\hat{\bu}/\partial\hat{\bm W}
\|
\leq
L_W
\).

For
\(
t\in[t_k,t_{k+1})
\),
the held quantity
\(
\hat{\bu}(
\bx(t_k),
\hat{\bm W}(t_k))
\)
is constant. Differentiating
\(
\bm e_u(t)
=
\hat{\bu}(
\bx(t_k),
\hat{\bm W}(t_k))
-
\hat{\bu}(
\bx(t),
\hat{\bm W}(t))
\)
therefore gives
\(
\dot{\bm e}_u
=
-
(\partial\hat{\bu}/\partial\bx)
\dot{\bx}
-
(\partial\hat{\bu}/\partial\hat{\bm W})
\dot{\hat{\bm W}}
\).
Hence,
\(
\|\dot{\bm e}_u\|
\leq
\|
\partial\hat{\bu}/\partial\bx
\|
\|\dot{\bx}\|
+
\|
\partial\hat{\bu}/\partial\hat{\bm W}
\|
\|\dot{\hat{\bm W}}\|
\leq
L_x\bar F
+
L_W\bar W
\).
Define
\(
\bar L_e
:=
L_x\bar F
+
L_W\bar W
>0
\).
Thus,
\(
\|\dot{\bm e}_u(t)\|
\leq
\bar L_e
\).

Since
\(
\bm e_u(t_k)=\bm0
\),
one has
\(
\bm e_u(t)
=
\int_{t_k}^{t}
\dot{\bm e}_u(\tau)\,d\tau
\),
and hence
\(
\|\bm e_u(t)\|
\leq
\int_{t_k}^{t}
\|\dot{\bm e}_u(\tau)\|\,d\tau
\leq
\bar L_e(t-t_k)
\).
Consequently,
\(
\|\bm e_u(\tau)\|
\leq
\bar L_e(\tau-t_k)
\)
for every
\(
\tau\in[t_k,t]
\),
and therefore
\(
\int_{t_k}^{t}
\|\bm e_u(\tau)\|^2\,d\tau
\leq
\bar L_e^2
\int_{t_k}^{t}
(\tau-t_k)^2\,d\tau
=
\frac{\bar L_e^2}{3}
(t-t_k)^3
\).
It follows that
\(
\sqrt{
(t-t_k)
\int_{t_k}^{t}
\|\bm e_u(\tau)\|^2\,d\tau
}
\leq
\frac{\bar L_e}{\sqrt{3}}
(t-t_k)^2
\).

By the definition of the composite event-error measure,
\(
\Xi_u(t)
=
\bar G\|\bm e_u(t)\|
+
\eta\bar G
\sqrt{
(t-t_k)
\int_{t_k}^{t}
\|\bm e_u(\tau)\|^2\,d\tau
}
\),
and hence
\(
\Xi_u(t)
\leq
\bar G\bar L_e(t-t_k)
+
\frac{\eta\bar G\bar L_e}{\sqrt{3}}
(t-t_k)^2
\).

At a nonterminal event instant
\(
t_{k+1}
\),
the triggering boundary is reached:
\(
\Xi_u(t_{k+1})
=
\sigma
[
C_1\mathcal J_x^\mu(t_{k+1})
+
C_2\mathcal J_x^\nu(t_{k+1})
]
+
\delta_{\mathrm{ET}}
\).
Since
\(
\sigma
[
C_1\mathcal J_x^\mu
+
C_2\mathcal J_x^\nu
]
\geq0
\),
it follows that
\(
\Xi_u(t_{k+1})
\geq
\delta_{\mathrm{ET}}
\).

Let
\(
\tau_k
:=
t_{k+1}-t_k
\).
Using the preceding upper bound on
\(
\Xi_u
\),
one obtains
\(
\delta_{\mathrm{ET}}
\leq
\bar G\bar L_e\tau_k
+
\frac{\eta\bar G\bar L_e}{\sqrt{3}}
\tau_k^2
\).
Define
\(
a_e
:=
\bar G\bar L_e>0
\)
and
\(
b_e
:=
\eta\bar G\bar L_e/\sqrt{3}>0
\).
Then
\(
b_e\tau_k^2
+
a_e\tau_k
-
\delta_{\mathrm{ET}}
\geq0
\).
Since
\(
\tau_k>0
\),
the positive root gives
\begin{equation}
    t_{k+1}-t_k
    \geq
    \frac{
        \sqrt{
            a_e^2
            +
            4b_e\delta_{\mathrm{ET}}
        }
        -
        a_e
    }{
        2b_e
    }
    >0.
    \label{eq:explicit_miet_composite}
\end{equation}
Thus, every nonterminal inter-event interval admits a uniform
strictly positive lower bound, so infinitely many events cannot
accumulate in a finite time interval. Hence, Zeno behavior is
excluded.
\end{proof}



\subsection{Simulation Setup}
\label{subsec:simulation_setup}

Consider the nonlinear control-affine plant
\(
\dot{\bx}
=
\mathbf f(\bx)
+
\bg(\bx)\bu
\),
where
\(
\mathbf f(\bx)
=
[f_1(\bx),f_2(\bx),f_3(\bx)]^\top
\)
with
\(
f_1(\bx)
=
-\tanh(x_1+x_2)-0.4x_1^3
\),
\(
f_2(\bx)
=
-\sin(x_2+x_3)-0.3x_1x_2
\),
and
\(
f_3(\bx)
=
-\frac{x_3}{1+x_3^2}
-0.2\tanh(x_1x_3)
\).
Moreover,
\(
\bg(\bx)=\bm I_3
\),
\(
\bx=[x_1,x_2,x_3]^\top\in\mathbb R^3
\),
and
\(
\bu=[u_1,u_2,u_3]^\top\in\mathbb R^3
\).
The identified drift is
\(\hat{\bfv}(\bx)=\hat{\bW}_f^\top\bsigma_f(\bx)\),
where
\(
\hat{\bm W}_f
\)
is obtained from
\eqref{eq:weight_identification},
\(
\bm R=\frac{1}{2}\bm I_3
\).

The desired value function is chosen as
\(
V_d(\bx)
=
\frac{1}{2}\bx^\top\bx
=
\frac{1}{2}
(x_1^2+x_2^2+x_3^2)
\),
and
\(
r^2(\bx)
:=
x_1^2+x_2^2+x_3^2
\).
Using the identified drift, define
\(
\hat F_V(\bx)
:=
\nabla V_d^\top(\bx)
\hat{\mathbf f}(\bx)
\).
To impose the fixed-time decay
\(
\dot V_d
=
-2V_d^{3/4}
-
V_d^{3/2}
\),
define
\(
\hat A(\bx)
:=
\hat F_V(\bx)
+
2V_d^{3/4}(\bx)
+
V_d^{3/2}(\bx)
\).

The inverse-optimal running cost is selected as
\(
L(\bx,\bu)
=
L_1(\bx)\bu
+
L_2(\bx)
+
\bu^\top\bm R\bu
\),
where, for
\(
\bx\neq\bm0
\),
\(
L_1(\bx)
=
\left(
\frac{\hat A(\bx)}
{r^2(\bx)}
-1
\right)
\bx^\top
\)
and
\(
L_2(\bx)
=
-\hat F_V(\bx)
+
\frac{\hat A^2(\bx)}
{2r^2(\bx)}
\).

Accordingly, the associated inverse-optimal feedback policy is
\(
\bu^*(\bx)
=
-\frac{\hat A(\bx)}
{r^2(\bx)}
\bx
\),
for
\(
\bx\neq\bm0
\),
with
\(
\bu^*(\bm0)=\bm0
\).
Under this feedback law,
\(
\dot V_d(\bx)
=
-2V_d^{3/4}(\bx)
-
V_d^{3/2}(\bx)
\),
which establishes the desired fixed-time convergence property.

\subsection{Simulation Results}
\label{subsec:simulation_results}

\begin{figure}[H]
    \centering
    \includegraphics[
        width=0.88\columnwidth,
        height=0.20\textheight
    ]{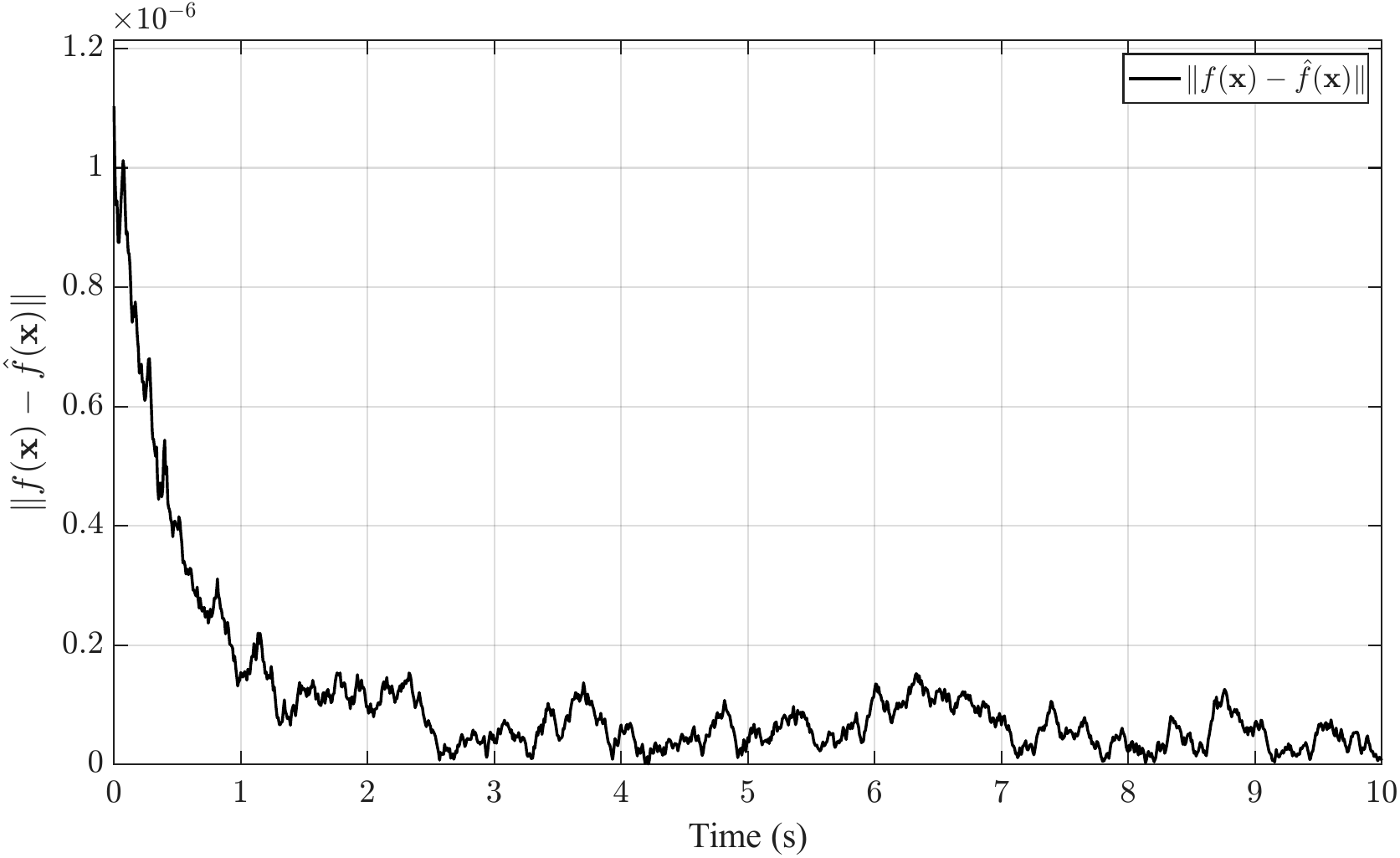}
    \caption{Identification performance.}
    \label{fig:identification_performance}
\end{figure}

\begin{figure}[H]
    \centering
    \includegraphics[
        width=0.88\columnwidth,
        height=0.14\textheight
    ]{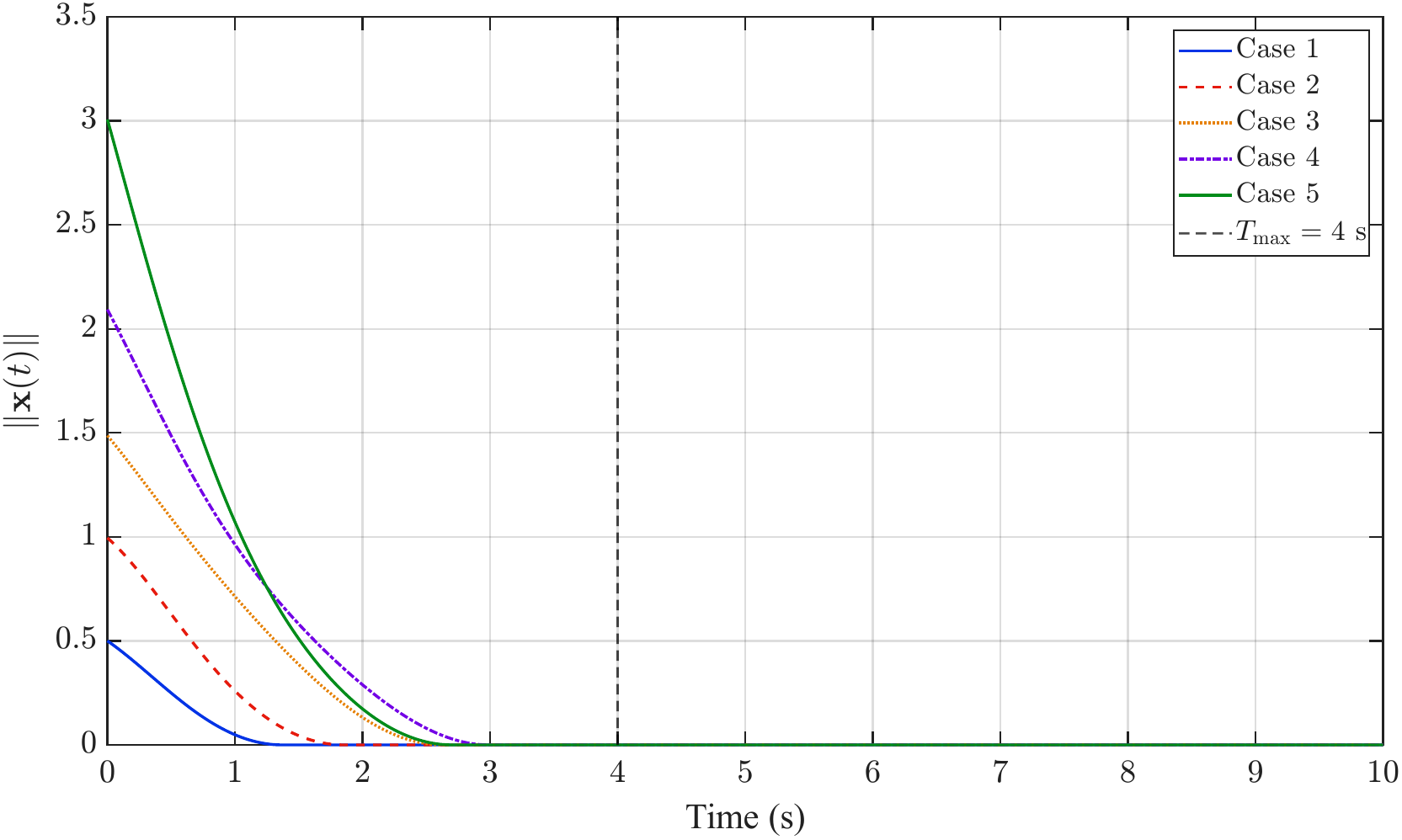}
    \caption{Closed-loop responses under different initial conditions.}
    \label{fig:closed_loop_response}
\end{figure}

\begin{figure}[H]
    \centering
    \includegraphics[
        width=0.88\columnwidth,
        height=0.14\textheight
    ]{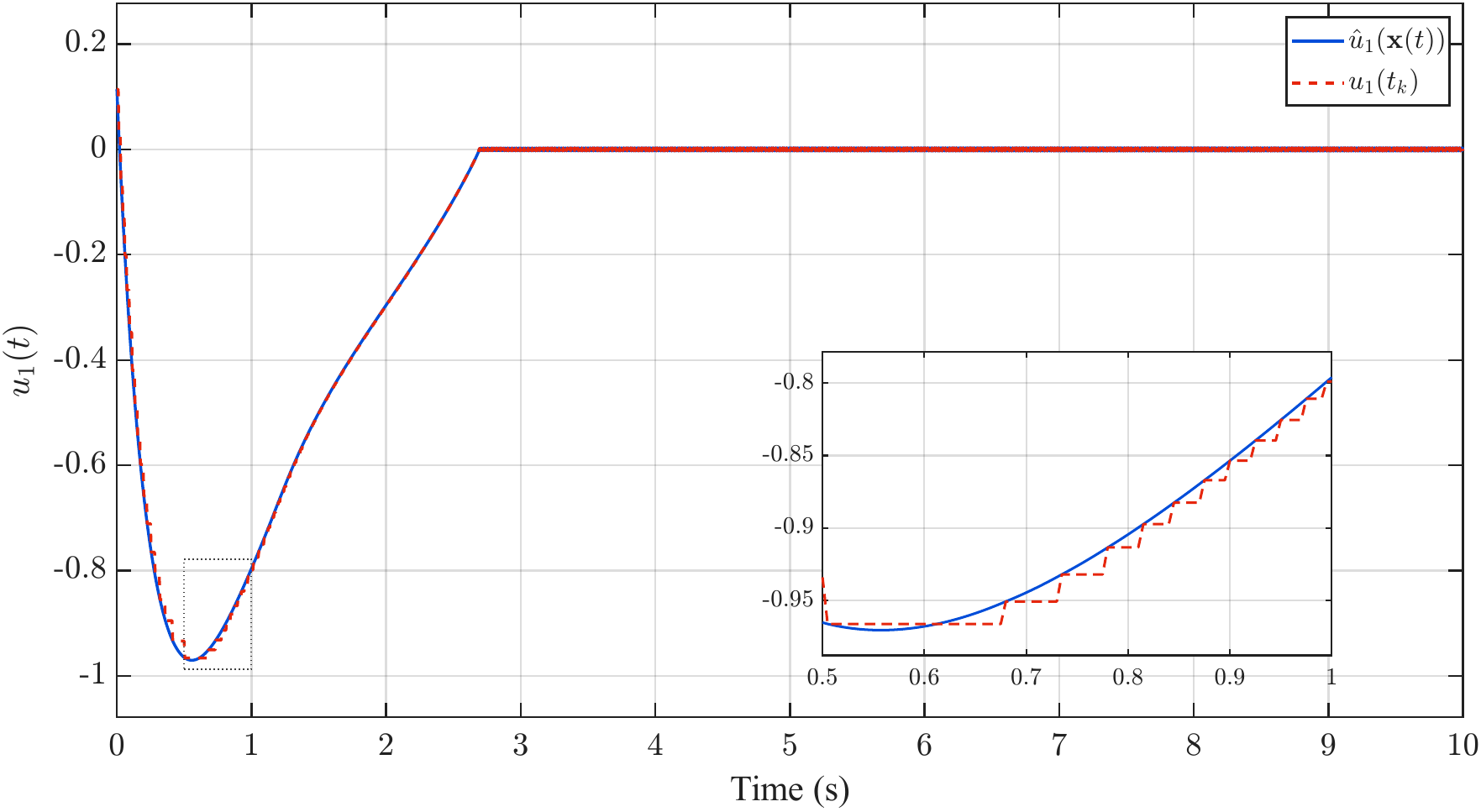}
    \caption{Event-triggered control input and continuously evaluated control policy.}
    \label{fig:event_control_input}
\end{figure}

\begin{figure}[H]
    \centering
    \includegraphics[
        width=0.88\columnwidth,
        height=0.14\textheight]
    {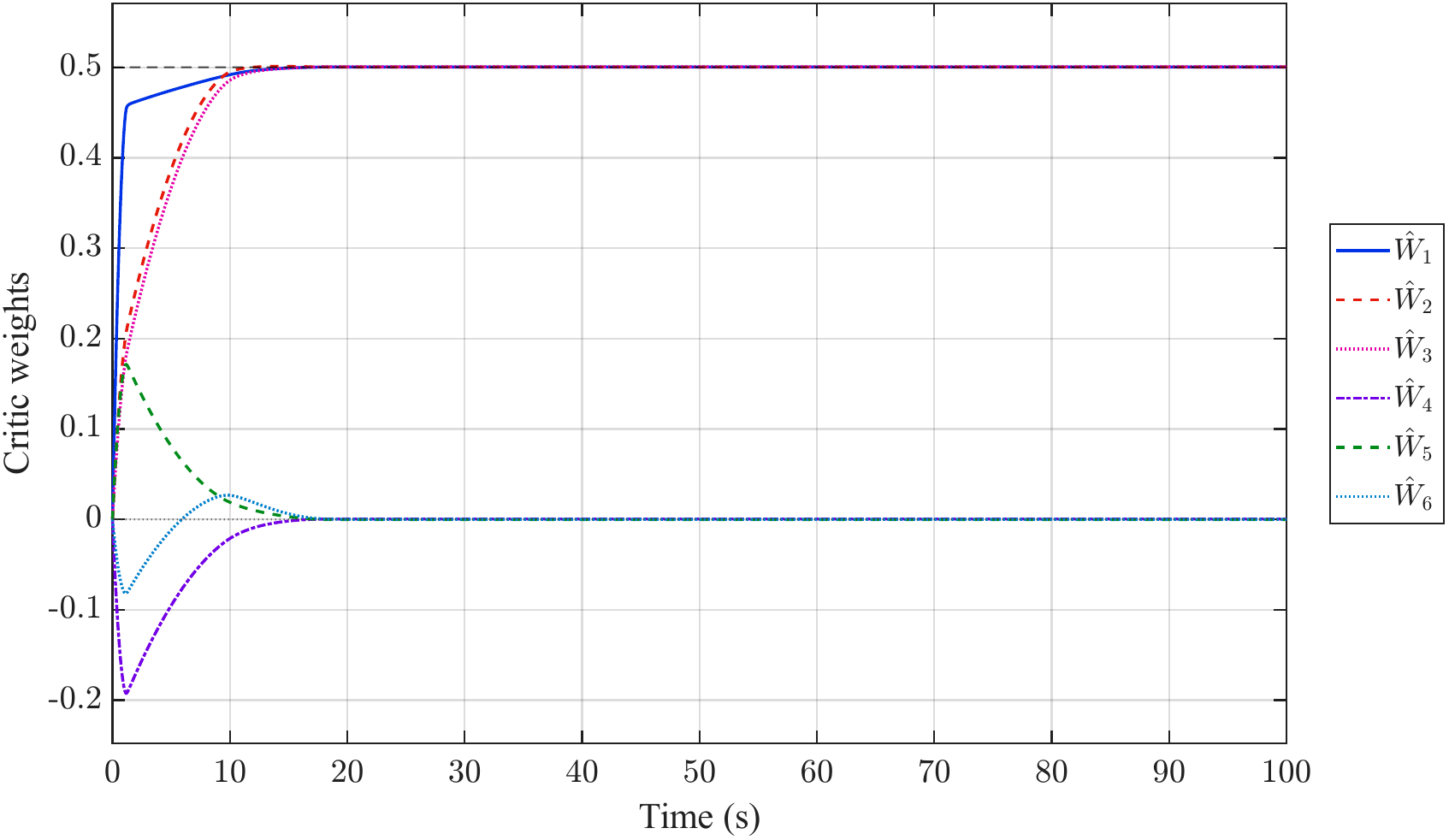}
    \caption{Evolution of the critic weights.}
    \label{fig:critic_weights}
\end{figure}





\begin{figure}[H]
    \centering
    \includegraphics[
        width=0.88\columnwidth,
        height=0.14\textheight
    ]{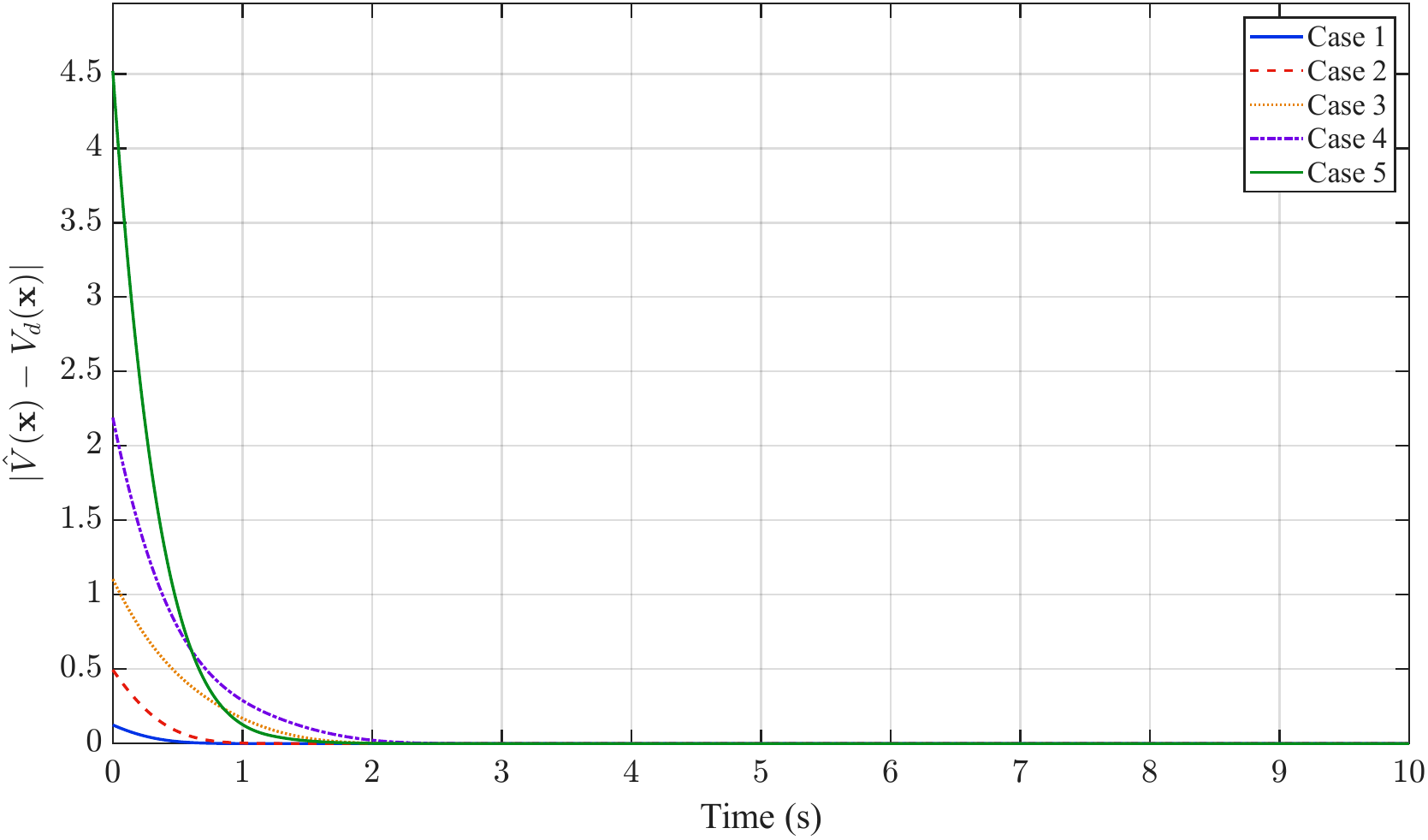}
    \caption{Value-function approximation error under different initial conditions.}
    \label{fig:value_function_error}
\end{figure}

\begin{figure}[H]
    \centering
    \includegraphics[
        width=0.88\columnwidth,
        height=0.14\textheight
    ]{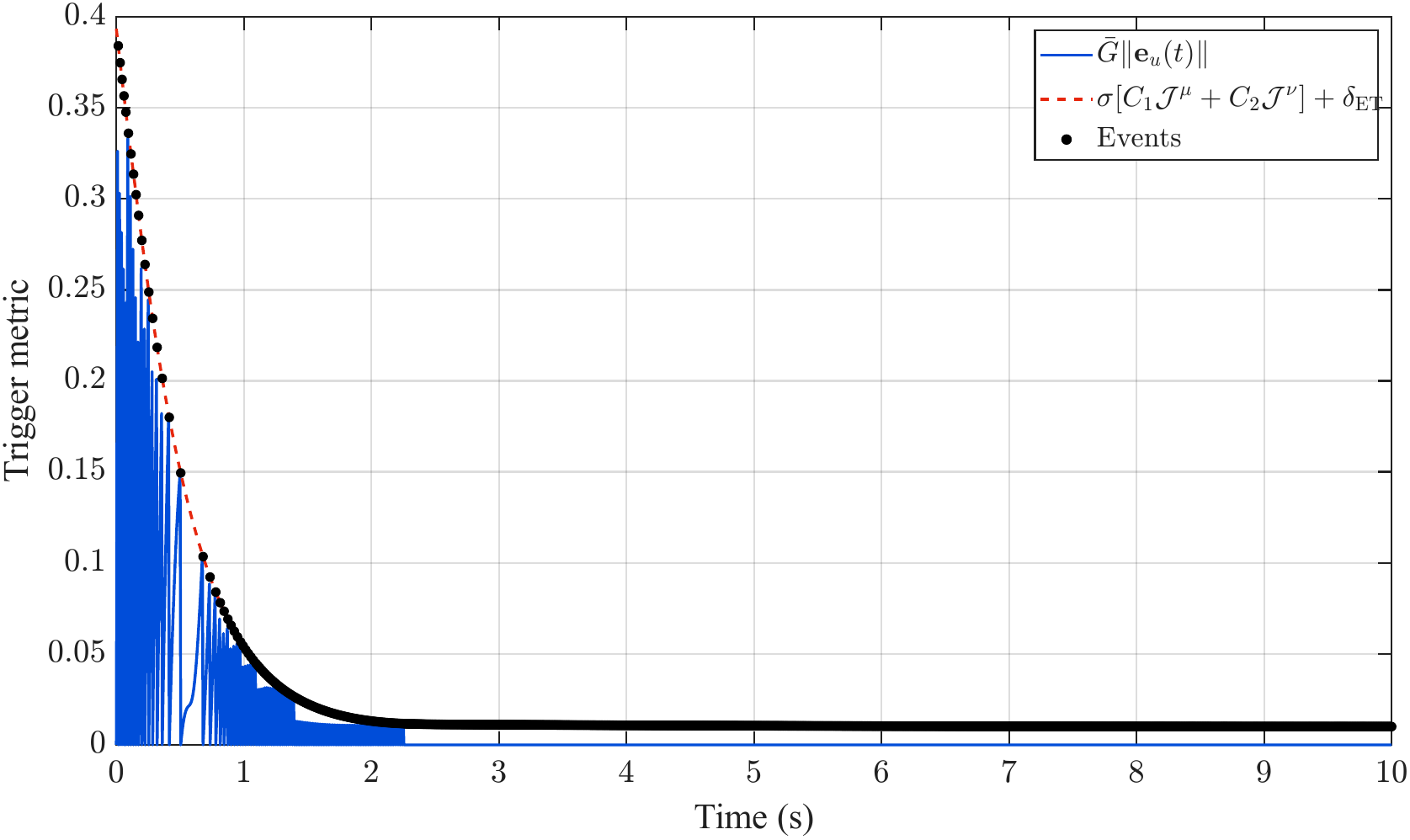}
    \caption{Evolution of the event-triggering condition.}
    \label{fig:trigger_condition}
\end{figure}

\begin{figure}[H]
    \centering
    \includegraphics[
        width=0.88\columnwidth,
        height=0.14\textheight
    ]{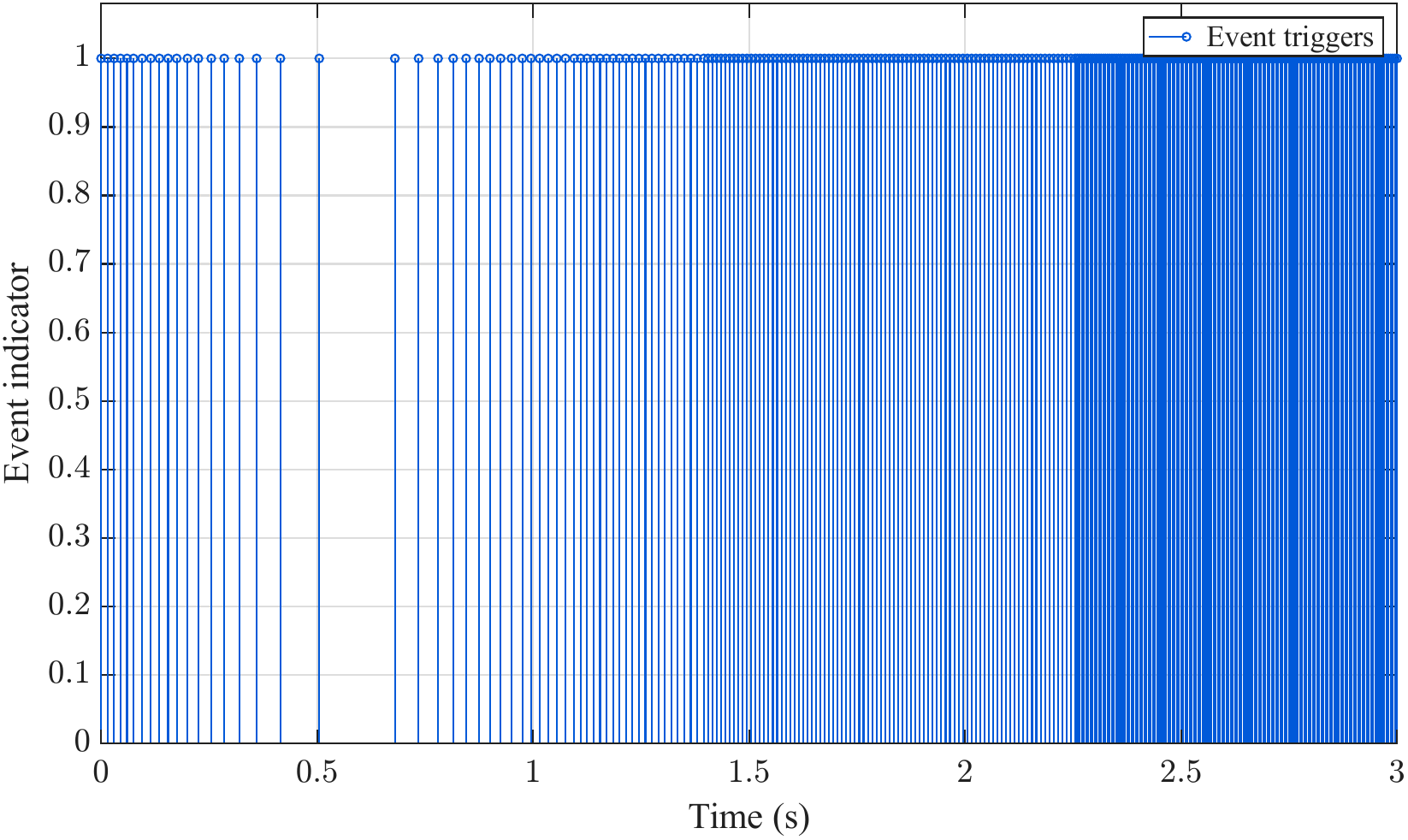}
    \caption{Event-triggering instants.}
    \label{fig:event_instants}
\end{figure}

\begin{figure}[H]
    \centering
    \includegraphics[
        width=0.88\columnwidth,
        height=0.14\textheight
    ]{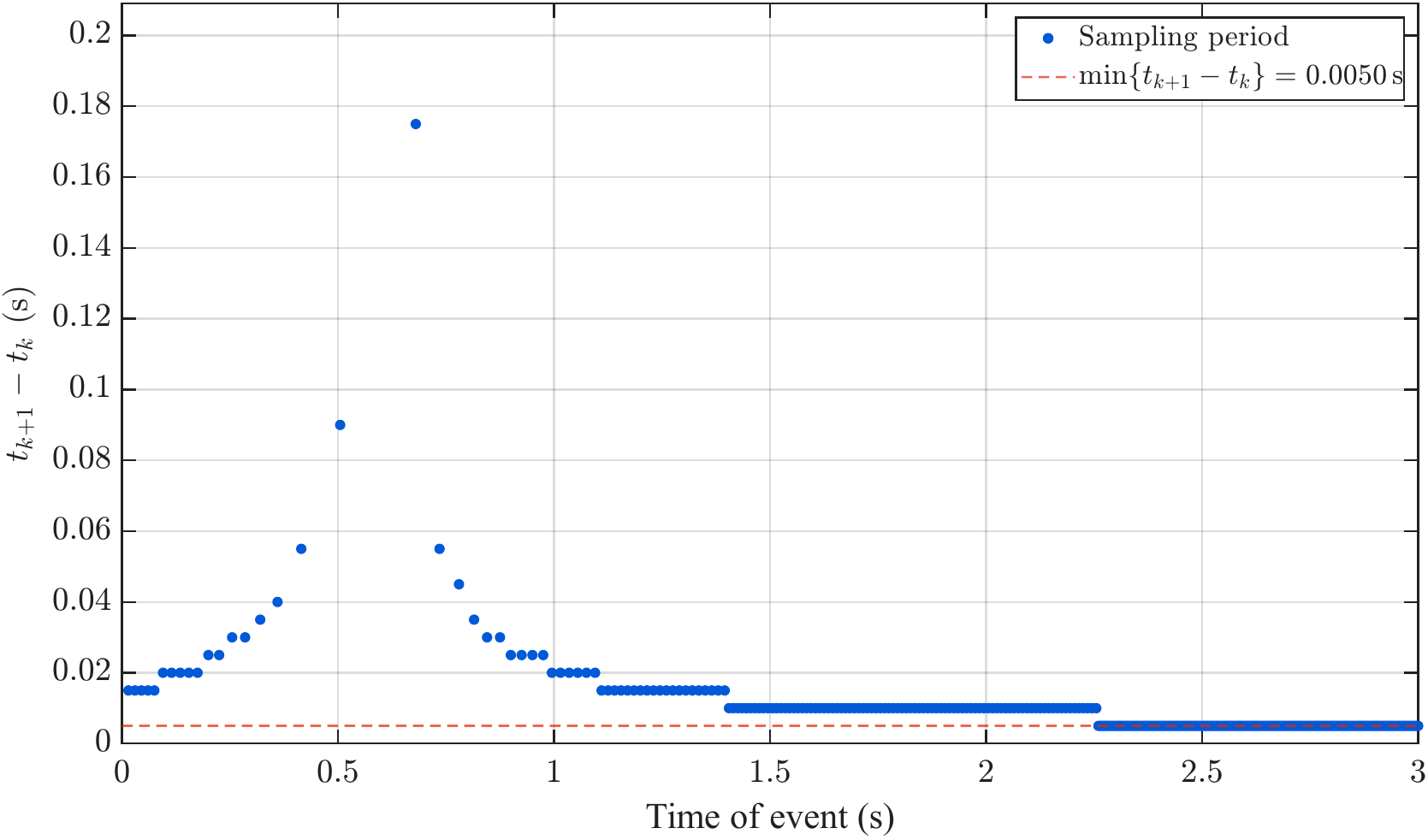}
    \caption{Inter-event times generated by the event-triggered implementation.}
    \label{fig:inter_event_times}
\end{figure}


The numerical results are consistent with the theoretical developments. 
Fig.~\ref{fig:identification_performance} shows that the drift-reconstruction error decreases from approximately \(1.1\times10^{-4}\) to the order of \(10^{-5}\) within \(2\)--\(3\,\mathrm{s}\), and remains at this level thereafter, confirming accurate finite-data identification before the inverse-optimal and IRL stages.

Fig.~\ref{fig:closed_loop_response} illustrates the fixed-time behavior for five initial conditions with \(\|\bx(0)\|\) ranging from approximately \(0.5\) to \(3.0\). All trajectories reach a small neighborhood of the origin within about \(4\,\mathrm{s}\), including the largest initial condition, consistent with the theoretical bound \(T_{\max}=4\,\mathrm{s}\). The comparable settling times across different initial conditions further support the fixed-time property.

Fig.~\ref{fig:event_control_input} shows that the event-triggered sample-and-hold input closely tracks the continuously evaluated policy. The first input reaches approximately \(-0.97\) near \(0.5\,\mathrm{s}\) and approaches zero by about \(2.7\,\mathrm{s}\), while the inset shows a discrepancy of only the order of \(10^{-3}\), indicating a small event-induced implementation error.

In Fig.~\ref{fig:critic_weights}, the first three critic weights converge to approximately \(0.5\), while the remaining weights converge to zero, yielding
\(
\bW^*=[0.5,\,0.5,\,0.5,\,0,\,0,\,0]^\top
\),
which exactly represents
\(
V_d(\bx)=\frac12\bx^\top\bx
\)
under the selected basis. The main critic transient disappears within approximately \(15\)--\(20\,\mathrm{s}\), while Fig.~\ref{fig:value_function_error} shows that the value-function error becomes visually negligible within about \(2\,\mathrm{s}\) for all five cases.

Figs.~\ref{fig:trigger_condition}--\ref{fig:inter_event_times} show that the triggering quantities decay with the closed-loop and learning transients, while the observed inter-event times remain positive, ranging from approximately \(5\times10^{-3}\,\mathrm{s}\) to \(1.7\times10^{-1}\,\mathrm{s}\). Since the minimum equals the integration step, the simulation only indicates the absence of numerical event accumulation; the uniform positive lower bound required for Zeno exclusion follows analytically from Lemma~\ref{lem:non_zeno_et}.


\section{Conclusion}
\label{sec:conclusion}

This paper developed an event-triggered fixed-time integral reinforcement
learning framework for optimal control of unknown nonlinear systems.
The unknown drift was reconstructed from integral data and used to build
an inverse-optimal fixed-time control problem. A critic-only learning law
with experience replay achieved practical fixed-time convergence without
requiring persistent excitation, while the event-triggered implementation
preserved practical fixed-time closed-loop stability and excluded Zeno
behavior. Simulation results verified the effectiveness of the proposed
framework.

\section*{Acknowledgment}
We acknowledge Ho Chi Minh City University of Technology (HCMUT), VNU-HCM for supporting this study.


\bibliographystyle{IEEEtran}
\bibliography{references}

@article{Doya2000,
  author  = {Kenji Doya},
  title   = {Reinforcement Learning in Continuous Time and Space},
  journal = {Neural Computation},
  volume  = {12},
  number  = {1},
  pages   = {219--245},
  year    = {2000},
  doi     = {10.1162/089976600300015961}
}

@article{Vrabie2009,
  author  = {Draguna Vrabie and Octavian Pastravanu and Murad Abu-Khalaf and Frank L. Lewis},
  title   = {Adaptive Optimal Control for Continuous-Time Linear Systems Based on Policy Iteration},
  journal = {Automatica},
  volume  = {45},
  number  = {2},
  pages   = {477--484},
  year    = {2009},
  doi     = {10.1016/j.automatica.2008.08.017}
}

@article{Vamvoudakis2010,
  author  = {Kyriakos G. Vamvoudakis and Frank L. Lewis},
  title   = {Online Actor--Critic Algorithm to Solve the Continuous-Time Infinite Horizon Optimal Control Problem},
  journal = {Automatica},
  volume  = {46},
  number  = {5},
  pages   = {878--888},
  year    = {2010},
  doi     = {10.1016/j.automatica.2010.02.018}
}

@article{Polyakov2012,
  author  = {Andrey Polyakov},
  title   = {Nonlinear Feedback Design for Fixed-Time Stabilization of Linear Control Systems},
  journal = {IEEE Transactions on Automatic Control},
  volume  = {57},
  number  = {8},
  pages   = {2106--2110},
  year    = {2012},
  doi     = {10.1109/TAC.2011.2179869}
}

@article{Modares2014,
  author  = {Hamidreza Modares and Frank L. Lewis and Mohammad-Bagher Naghibi-Sistani},
  title   = {Integral Reinforcement Learning and Experience Replay for Adaptive Optimal Control of Partially-Unknown Constrained-Input Continuous-Time Systems},
  journal = {Automatica},
  volume  = {50},
  number  = {1},
  pages   = {193--202},
  year    = {2014},
  doi     = {10.1016/j.automatica.2013.09.043}
}

@article{Han2024ETIRL,
  author  = {Xiumei Han and Xudong Zhao and Ding Wang and Bohui Wang},
  title   = {Event-Triggered-Based Online Integral Reinforcement Learning for Optimal Control of Unknown Constrained Nonlinear Systems},
  journal = {International Journal of Control},
  volume  = {97},
  number  = {2},
  pages   = {213--225},
  year    = {2024},
  doi     = {10.1080/00207179.2022.2137852}
}

@article{Hong2025DETADP,
  author  = {Yishen Hong and Shan Xue and Derong Liu and Yonghua Wang},
  title   = {Dynamic Event-Triggering Adaptive Dynamic Programming for Robust Stabilization of Partially Unknown Nonlinear Systems},
  journal = {Neurocomputing},
  volume  = {628},
  pages   = {129673},
  year    = {2025},
  doi     = {10.1016/j.neucom.2025.129673}
}

@article{Zhang2026ETFixedTime,
  author  = {Xingyi Zhang and Hui Ma and Hongru Ren and Qi Zhou},
  title   = {Event-Triggered Optimal Tracking Control of Uncertain Nonlinear Systems with Fixed-Time Prescribed Performance},
  journal = {International Journal of Control},
  volume  = {99},
  number  = {2},
  pages   = {245--257},
  year    = {2026},
  doi     = {10.1080/00207179.2025.2512399}
}

@article{Wan2021CPS,
  author  = {Yan Wan and Tao Yang and Ye Yuan and Frank L. Lewis},
  title   = {Optimal Control and Learning for Cyber-Physical Systems},
  journal = {International Journal of Robust and Nonlinear Control},
  volume  = {31},
  number  = {6},
  pages   = {1799--1802},
  year    = {2021},
  doi     = {10.1002/rnc.5442}
}

@article{Lee2015CPS,
  author  = {Edward A. Lee},
  title   = {The Past, Present and Future of Cyber-Physical Systems:
             A Focus on Models},
  journal = {Sensors},
  volume  = {15},
  number  = {3},
  pages   = {4837--4869},
  year    = {2015},
  doi     = {10.3390/s150304837}
}

@article{Jimenez2017Predefined,
  author  = {Esteban Jim{\'e}nez-Rodr{\'i}guez
             and Juan D. S{\'a}nchez-Torres
             and Alexander G. Loukianov},
  title   = {On Optimal Predefined-Time Stabilization},
  journal = {International Journal of Robust and Nonlinear Control},
  volume  = {27},
  number  = {17},
  pages   = {3620--3642},
  year    = {2017},
  doi     = {10.1002/rnc.3757}
}

@inproceedings{Kokolakis2023FixedTime,
  author    = {Nick-Marios T. Kokolakis
               and Kyriakos G. Vamvoudakis
               and Wassim M. Haddad},
  title     = {Fixed-Time Learning for Optimal Feedback Control},
  booktitle = {ASME 2023 International Design Engineering Technical
               Conferences and Computers and Information in Engineering
               Conference},
  year      = {2023},
  doi       = {10.1115/DETC2023-117007}
}

@article{Kokolakis2025Predefined,
  author  = {Nick-Marios T. Kokolakis
             and Kyriakos G. Vamvoudakis
             and Wassim M. Haddad},
  title   = {Predefined-Time Reinforcement Learning for Optimal
             Feedback Control},
  journal = {Automatica},
  volume  = {179},
  pages   = {112421},
  year    = {2025},
  doi     = {10.1016/j.automatica.2025.112421}
}

@article{AbuKhalaf2005,
  author  = {Murad Abu-Khalaf and Frank L. Lewis},
  title   = {Nearly Optimal Control Laws for Nonlinear Systems with
             Saturating Actuators Using a Neural Network HJB Approach},
  journal = {Automatica},
  volume  = {41},
  number  = {5},
  pages   = {779--791},
  year    = {2005},
  doi     = {10.1016/j.automatica.2004.11.034}
}

@misc{VuDoan2026PredefinedResilientIRL,
  author        = {Tien Dat Vu and Minh Doan},
  title         = {Predefined-Time Resilient Integral Reinforcement Learning for Input-Constrained Unknown Nonlinear Systems Under FDI Attacks and Disturbances: A Fully Data-Driven Approach},
  year          = {2026},
  eprint        = {2609.11815},
  archivePrefix = {arXiv},
  primaryClass  = {eess.SY},
  note          = {arXiv:2609.11815}
}

\end{document}